\documentclass[11pt]{article}
\usepackage{bbold}
\usepackage{arxiv}
\usepackage{booktabs}
\usepackage{enumitem}
\usepackage{amsmath}
\usepackage{amsfonts}
\usepackage{amssymb}
\usepackage{mathtools}
\usepackage{graphicx}
\usepackage{color}
\usepackage{fullpage}
\usepackage[round]{natbib}

\newcommand{\calG}{{\cal G}}

\newcommand{\calI}{{\cal I}}

\newcommand{\calS}{{\cal S}}

\newcommand{%
  \sjacobian}{\mathbf{SJ}}
\newcommand{%
  \jacobian}{\mathbf{J}}
\newcommand{%
  \hessian}{\mathbf{H}}
\newcommand{\bbone}{\mathbb{1}}

\newcommand{\supp}{\textnormal{supp}}
\newcommand{\spr}{{\mathrm{SPR}}}
\usepackage{comment}
\usepackage{stmaryrd}

\usepackage{sectsty}
\sectionfont{\large}
\subsectionfont{\normalsize}
\subsubsectionfont{\normalsize}
\paragraphfont{\normalsize}

\RequirePackage{algorithm}
\RequirePackage{algorithmic}
\usepackage[linktocpage=true]{hyperref}
\hypersetup{
  colorlinks=true,
  citecolor=blue,
  linkcolor=magenta,
}
\usepackage[capitalize,noabbrev]{cleveref}
\crefname{assumption}{Assumption}{Assumptions}
\usepackage{amsthm}
\usepackage{thmtools}
\usepackage{thm-restate}

\declaretheorem[name=Theorem, numberwithin=section]{theorem}

\declaretheorem[name=Lemma, numberlike=theorem]{lemma}

\declaretheorem[name=Definition, numberlike=theorem]{definition}

\declaretheorem[name=Corollary, numberlike=theorem]{cor}
\declaretheorem[name=Remark, numberlike=theorem]{remark}

\declaretheorem[name=Example, numberlike=theorem]{example}
\declaretheorem[name=Assumption, numberlike=theorem]{assumption}

\usepackage{multicol}
\usepackage{subcaption}

\newcommand\todo[1]{\noindent\textcolor{red}{#1}}
\newcommand{\argmin}{\operatornamewithlimits{argmin}}
\newcommand{\argmax}{\operatornamewithlimits{argmax}}

\renewcommand{\headeright}{}
\renewcommand{\undertitle}{}
\newcommand\blfootnote[1]{%
  \begingroup
  \renewcommand\thefootnote{}\footnote{#1}%
  \addtocounter{footnote}{-1}%
  \endgroup
}

\title{
  Computing Equilibria in Games with Stochastic Action Sets}

\usepackage{footmisc}
\usepackage{authblk}

\author{Thomas Schwarz}
\author{Jiaru Li}
\author{Ryann Sim}
\author{Chun Kai Ling}
\affil{School of Computing, National University of Singapore}

\date{}

\begin{document}

\maketitle
\blfootnote{Contact: \texttt{tschwarz@comp.nus.edu.sg},
\texttt{jiaru.li@u.nus.edu},
\texttt{ryann.sim@nus.edu.sg}, \texttt{chunkail@nus.edu.sg}}

\begin{abstract}
The study of learning in games typically assumes that each player always has access to all of their actions. However, in many practical scenarios, players' available actions might be restricted due to exogenous stochasticity. To model this setting, for a game $\mathcal{G}_{\mathrm{orig}}$ with action set $A_i$ for each player $i$, we introduce the corresponding Game with Stochastic Action Sets (GSAS) which is parametrized by a probability distribution over the players' set of possible action subsets $\mathcal{S}_i\subseteq 2^{A_i}\setminus\{\varnothing\}$. In a GSAS, players' strategies and Nash equilibria (NE) admit prohibitively large representations, and existing algorithms for NE computation scale poorly. Under the assumption that action availabilities are independent between players, we show that NE in two-player zero-sum (2p0s) GSAS can be approximately represented by a compact vector of size $\vert A_i\vert$, overcoming the naïve exponential-sized representation. Computationally, we introduce an algorithm that minimizes ranking regret, converging to NE with high probability in 2p0s-GSAS with rate $O(\sqrt{\vert A_i\vert\log\vert A_i\vert/T})$ for time horizon $T$. Finally, using the iterates of our algorithm, we develop a stochastic approximation procedure to recover compactly represented NE.
\end{abstract}





\section{Introduction}\label{sec:intro}
A common assumption in game theory is that the players always have access to all of their actions. However, in many multi-agent systems, a player's actions might be randomly restricted at certain times. For instance, a trader might only have access to a subset of options on any given day due to exogenous constraints outside their control. A naval ship might be unable to enter or patrol part of its usual operating area due to the presence of civilians or inclement weather. In these scenarios, a player has access to only a subset of their action set and must choose among the available actions, which could change each time the game is played. 


To formally model this phenomenon, we introduce the class of Games with Stochastic Action Sets (GSAS). 
A GSAS $\mathcal{G}$ begins with a `base' normal-form game $\calG_{\mathrm{orig}}$ with finite action sets $A_i$ for each player $i$. We further equip $\calG$ with a probability distribution $\rho \in \Delta(\calS)$, where $\calS \coloneqq \bigtimes_i\calS_i$ and each $\calS_i$ is the set of all possible action subsets $\calS_i \subseteq 2^{A_i}\setminus\{\varnothing\}$ for each player $i$.
At each round, Nature draws a joint availability set $S\in\calS$ from $\rho$. Each player observes only their own availability set $S_i\subseteq A_i$ and is restricted to selecting an action $a_i\in S_i$. 

In this paper, we are interested in obtaining Nash equilibria (NE) of GSAS, which are given by strategy profiles from which no player has an incentive to unilaterally deviate. 
However, existing strategy formalisms require that players select, for each possible action subset, a distribution over available actions to play. 
For instance, formulating a GSAS as a Bayesian game~\citep{harsanyi1968games} requires modeling action availabilities as types, where strategies are mappings from types to action distributions.
Furthermore, since some actions may be unavailable,
the NE structure of a GSAS $\calG$ can differ significantly from $\calG_{\mathrm{orig}}$ (cf.~\cref{ex:matching_pennies}).



This raises a key challenge in understanding the properties of NE in GSAS. In particular, the naïve specification of a player's strategy might require a mapping from action availabilities to distributions over actions, whose domain could be exponentially large in the worst case. Thus, even \emph{specifying} a NE could be challenging, raising concerns about NE as an implementable solution concept and motivating a study of conditions under which NE can be \emph{compactly represented} in GSAS. This type of result has several seminal precedents: Kuhn's theorem~\citep{kuhn1953extensive} showed that behavioral strategies (as opposed to  exponentially sized mixed strategies) suffice to capture NE in extensive-form games of perfect recall; \citet{shapley1953stochastic} showed that every two-player zero-sum Markov game admits a Markovian equilibrium.

We next turn our attention to the complementary task of \textit{computing} NE in the special but important case of two-player zero-sum (2p0s) GSAS.
A standard paradigm for equilibrium computation in normal-form games utilizes a canonical connection between no-regret learning and game-theoretic equilibria, which states that no-regret learners converge in time-average to approximate equilibria~\citep{cesa2006prediction}. This connection has led to the design of decentralized algorithms that can efficiently compute NE in standard two-player zero-sum games; see, e.g., ~\citep{freund1999adaptive,rakhlin2013optimization,daskalakis2011near,syrgkanis2015fast,daskalakis2021near}. However, when attempting to solve a 2p0s-GSAS directly, the connection between `standard' no-external-regret learning and equilibria breaks down since players' action sets can vary between rounds and the competing action may not be available. To circumvent this, one could also consider special cases of GSAS that allow for the application of existing techniques. For instance, when $\rho$ is fully known to the solver, the game can be expanded into a Bayesian game.
Nevertheless, even in this special case, existing solvers do not scale well since the number of possible action sets could be large.
This motivates the design of computational techniques
that can efficiently compute \emph{compact representations} of NE, 
even without specifying $\rho$ explicitly.

\paragraph{Our contributions.}
The primary contribution of our work is the design and analysis of a complete procedure that computes compact representations of approximate NE in 2p0s-GSAS. 
This requires three key aspects.
First, we \textbf{formalize and study properties of equilibria} in 2p0s-GSAS under the assumption that players' action availabilities are independent. We show that players' strategy sets can be restricted to the set of `implementable' marginals which contains all equilibria of the GSAS. We also establish that any implementable equilibrium marginal can be approximately represented as a vector $w$ of size $\vert A_i\vert$. Second, 
we introduce an algorithm called CAT-FTPL that \textbf{minimizes a suitable notion of regret} called \textit{ranking regret} in GSAS, and converges to approximate NE in 2p0s-GSAS. We prove a sublinear bound on expected pathwise ranking regret and derive a corresponding high-probability upper-tail bound.
Third, using \emph{only} the CAT-FTPL iterates, we introduce a stochastic approximation-based method to \textbf{extract a compact vector} $w$ that represents an approximate NE. 
Finally, we investigate the empirical efficacy of our proposed method in 2p0s-GSAS.
Our experimental results demonstrate the scalability of our method compared to standard game solvers and show that it exhibits approximate convergence to compact equilibria in large games.

\section{Related Work}\label{sec:related}
Due to space constraints, in this section we discuss related work that directly concerns our techniques and analysis. Additional discussions on related work and on the connection to Bayesian games are deferred to~\cref{appsecs:related}.

\paragraph{Sleeping bandits and stochastic actions.}
Our work is closely related to the \emph{sleeping bandit} setting in the multi-armed bandit literature, which studies regret minimization when arms are available stochastically or adversarially~\citep{auer2002finite,blum2007external,kanade2009sleeping,kleinberg2010regret,kanade2014learning,saha2020improved,nguyen2024near,schneider2023optimal}. These techniques have been applied to online combinatorial optimization~\citep{neu2014online,kale2016hardness} and reinforcement learning~\citep{drago2025sleeping}. Finally,~\cite{boutilier2018planning} studied MDPs with stochastic action availabilities, obtaining efficient algorithms. In our work, we focus on computing compact Nash equilibria in a multi-agent setting, which requires additional game-theoretic insights beyond those studied in prior works.

\paragraph{Games with exponential action sets.} 
\citet{panageas2023semi} and 
\citet{dong2023taming} studied semi-bandit learning in congestion games which admit exponentially large action sets, while~\citet{farina2022kernelized} and \citet{kontogiannis2025efficient} studied regret minimization in \emph{polyhedral} games (with underlying combinatorial structure). The emphasis of these works is on deriving strong regret guarantees, whereas we additionally analyze efficient equilibrium \emph{representation}.
\paragraph{Connection to Bayesian games.}
By interpreting each player's availability set $S_i$ as their private type drawn from the marginal distribution induced by $\rho$, a GSAS can be viewed as a variant of Bayesian games~\citep{harsanyi1968games}. However, while a GSAS shares some structure with Bayesian games, they are distinct in several ways, introducing new challenges and motivating specialized algorithms: (i) classical Bayesian game solvers assume common knowledge of payoff functions, possible types, and a prior distribution over types, whereas a GSAS allows for $\rho$ to be unknown to the solver; and (ii) algorithms for computing Bayesian NE have rates that depend on the number of types~\citep{fujii2025bayes,dagan2024external,peng2024complexity}.

\section{Games with Stochastic Action Sets}\label{sec:prelims}
\paragraph{Notation.} Denote the $n$-dimensional nonnegative orthant by $\mathbb{R}^n_{\ge0}$. For a finite set $S$, 
let $\Delta(S)$ be the probability simplex 
$\{ x \in \mathbb{R}^{|S|}_{\ge0}\mid \sum_{s\in S} x(s) = 1 \}$, 
such that if $y \in \Delta(S), s \in S$, $y(s)$ is
the probability 
that item $s$ is selected. 

Consider an $n$-player normal/strategic-form game $\mathcal{G}_\mathrm{orig}$ with finite action set $A_i$, joint action space $A = A_1 \times \dots \times A_n$, and payoff functions scaled as $u_i: A \rightarrow [-1, 1]$ for player $i$. 
In line with prevailing conventions, we let $a=(a_1,\dots,a_n)$ be an action profile, and as shorthand $u_i(a) = u_i(a_1,\dots,a_n)$. We also denote by $-i$ the set of players other than $i$, such that $u_j(a'_i, a_{-i}) = u_j(a_1,\dots, a_i', \dots, a_n)$.

\begin{definition}[GSAS]
    Given a game $\mathcal{G}_\mathrm{orig} = (A_1, \dots, A_n, u_1, \dots, u_n)$, let $\mathcal{S}_i \subseteq 2^{A_i} \setminus \{ \varnothing \}$, let $\mathcal{S} = \mathcal{S}_1 \times \dots \times \mathcal{S}_n$, and let $\rho \in \Delta(\mathcal{S})$ be a distribution over elements $\mathcal{S}$, such that $\rho(S)$ gives the probability that the availability set $S$ is observed for $S \in \mathcal{S}$.
    Then, a normal-form \textbf{G}ame with \textit{\textbf{S}tochastic \textbf{A}ction \textbf{S}ets} (\emph{GSAS}) is given by the tuple $\mathcal{G} = (\mathcal{G}_\mathrm{orig}, \mathcal{S}, \rho)$.
\end{definition}
\begin{definition}[2p0s-GSAS]
 A two-player zero-sum GSAS (\emph{2p0s-GSAS}) $\mathcal{G}$ is one where $\mathcal{G}_\mathrm{orig}$ is two-player zero-sum, i.e., $n=2, u_1(a)=-u_2(a)$ for all action profiles $a \in A$.
\end{definition}

Informally, a GSAS proceeds as follows. At the start of the game, each player privately receives their availability set $S_i \in \mathcal{S}_i$ from Nature based on $\rho$. Each then plays an action $a_i \in S_i$ simultaneously and receives a payoff $u_i(a)$ based on the action profile $a \in A$. While each player $i$ observes their availability set $S_i$ and might have full knowledge of $\rho$,
they do not observe their opponents' availability sets $S_{-i}$.
In the remainder of the paper, unless otherwise stated, we make a technical assumption that facilitates our analysis going forward. A similar assumption is often made to facilitate price-of-anarchy analyses in Bayesian games~\citep{fujii2025bayes, roughgarden2015price, syrgkanis2013composable, syrgkanis2012bayesian}. 

\begin{assumption}\label{assumption:product}
    $\rho(S) = \prod_{i=1}^{n} \rho_i(S_i)$ for some probability distributions $\rho_i: \mathcal{S}_i \rightarrow [0, 1]$, i.e., the availability of actions is independent across players. 
\end{assumption}

A \textit{pure strategy} in a GSAS is a deterministic mapping $\pi_i: \mathcal{S}_i \rightarrow A_i$ where $\pi_i(S_i) \in S_i$. More generally, we define a \textit{mixed strategy} (or simply \textit{strategy}) for player $i$ as a mapping $\pi_i: \mathcal{S}_i \to \Delta(A_i)$ such that $\operatorname{supp}(\pi_i(S_i)) \subseteq S_i$. 
A player's strategy gives, for every possible observed subset of actions, a distribution of actions corresponding to the observed action subset\footnote{Mixed strategies in Bayesian games are classically defined by \textit{distributions over pure strategies}. It follows from classical results \citep{harsanyi1968games} that this is strategically equivalent (for most equilibrium computation purposes, including ours) to our simpler definition of per-type conditional distribution over actions.}. Note that even under~\cref{assumption:product}, the explicit representation of the $\pi_i$'s could be exponentially large, and dealing with this is a key contribution of our work. 

Given a joint availability set $S\in\mathcal{S}$, we denote by $\pi$ the joint strategy of all players and by $\pi(a \mid S) = \prod_{i \in [n]} \pi_i(a_i \mid S_i)$ the probability of an action profile $a$ for every $a\in A$.
Player $i$'s expected payoff is
\begin{equation}
    U_i(\pi) = \mathbb{E}_{S \sim \rho}\left[\mathbb{E}_{a \sim \pi(S)}[u_i(a)]\right],
\end{equation}
where the inner expectation is over the actions sampled independently according to each player's strategy given their action sets, and the outer expectation is over the availability set $S$ drawn from $\rho$. We also define the expected payoff of a specific action $a_i$ of player $i$ with respect to the ensemble of the opponents' strategies $\pi_{-i}$ by
\begin{equation}
    U_i(a_i; \pi_{-i}) = \mathbb{E}_{S \sim \rho}\left[\mathbb{E}_{a_{-i} \sim \pi_{-i}(S_{-i})}[u_i(a_i, a_{-i})]\right].
\end{equation}


\begin{definition}[$\epsilon$-Nash equilibrium ($\epsilon$-NE)]\label{defn:ne}
For $\epsilon\ge 0$, a strategy profile $\pi = (\pi_1, \ldots, \pi_n)$ is an \emph{$\epsilon$-Nash equilibrium} if no player can improve their expected payoff by more than $\epsilon$ by unilaterally deviating from $\pi$, i.e., for every player $i \in [n]$ and for every alternative strategy $\pi_i'$,
\begin{equation}
    U_i(\pi_i, \pi_{-i}) \ge U_i(\pi_i', \pi_{-i}) - \epsilon.
\end{equation}

A $0$-NE is a \emph{Nash equilibrium}, i.e., no player can strictly benefit by unilaterally deviating.
\label{def:ne}
\end{definition}
These definitions extend the classical normal-form Nash equilibrium to GSAS and are consistent with Bayesian games and Bayesian Nash equilibria~\citep{harsanyi1968games}.
While it is tempting to assume that solutions to GSAS may be easily extracted from the underlying $\mathcal{G}_\mathrm{orig}$, we show in the following example that this is generally \emph{untrue}.
Indeed, the presence of stochastic action availabilities for a player can significantly change the set of NE, even if other players do not face stochastic action sets. 

\begin{example}
    Let $\mathcal{G}_\mathrm{orig}$ be a matching pennies game with actions $H$ and $T$. Define a GSAS by setting $u_1(H,H)=u_1(T,T)=1$, $u_1(H,T)=u_1(T,H)=-1$, and $u_2=-u_1$. Suppose that $\mathcal{S}_2 = \{ \{ H, T \} \}$, i.e., all actions are always available for player 2, but $\mathcal{S}_1=\{ S_{11} = \{ H, T \}, S_{12} = \{ H \} \}$ with $\rho_1(S_{11}) = \lambda$ and $\rho_1(S_{12}) = 1-\lambda$. If $\lambda=1$, we have the regular matching pennies game, while for $\lambda=0$, player 1's only available action is $H$. The set of NE is as follows: 
    \begin{enumerate}[nosep, topsep=0pt, partopsep=0pt, parsep=0pt, itemsep=0pt, leftmargin=*, label=(\roman*)]
    \item $\lambda > 0.5$: player 1 sets $\pi_1(T \mid S_{11}) = {0.5}/{\lambda}$ which ``effectively'' plays $T$ with probability $0.5$, while player $2$ plays uniformly at random: this is essentially the standard matching pennies; 
    \item $0<\lambda < 0.5$: player 1 sets $\pi_1(T\mid S_{11}) = 1$ and player 2 plays $T$ deterministically; and 
    \item $\lambda =0.5$: player 1 sets $\pi_1(T\mid S_{11})= 1$ and player 2 plays $T$ with probability in $[0.5, 1]$.
    \end{enumerate}
    
    Intuitively, when $0<\lambda < 0.5$, player 1 can never play $T$ frequently enough and player 2 takes advantage of this by playing $T$ deterministically, while player $1$ plays $T$ whenever possible. When $\lambda > 0.5$, player 1 compensates by playing $T$ with higher probability when $S_{11}$ is offered, i.e., $T$ is available, since they must play $H$ all the time if $S_{12}$ is drawn.
\label{ex:matching_pennies}
\end{example}


\section{Properties of Equilibria in GSAS}\label{sec:gsasproperties}

For a player $i$, the marginal distribution over actions $a_i \in A_i$ induced by any of their strategies $\pi_i$ is 
$\mathbb{P} \left[ a_i ; \rho_i, \pi_i \right] =
    \textstyle\sum_{S_i \in \mathcal{S}_i}  \left( \rho_i(S_i)  
     \pi_i(a_i\mid S_i) \bbone\{ a_i \in S_i \}\right)$. 
A crucial observation is that since action availabilities are independent by~\cref{assumption:product}, player $i$'s payoff can be expressed in terms of $\mathbb{P}[a_j; \rho_j, \pi_j]$ for all players $j$, rather than the much larger $\pi_j$'s (derivation in~\cref{appsec:eqn4deriv}):
\begin{align}
   U_i(\pi) = \sum_{a \in A} u_i(a) \prod_{j \in [n]} \mathbb{P} \left[ a_j ; \rho_j, \pi_j \right].
    \label{eq:depend-only-marginals}
\end{align}


\begin{definition} \label{def:marginal}
Let $\mu_i \in \Delta(A_i)$ be a probability distribution over player $i$'s possible actions. We say that $\mu_i$ is an \emph{implementable marginal} if there exists a strategy $\pi_i:\mathcal{S}_i \to \Delta(A_i)$ such that 
$\forall a_i \in A_i$, $\mu_i(a_i) = \mathbb{P}[a_i; \rho_i, \pi_i]$. In this case, we also say that $\pi_i$ \emph{implements} $\mu_i$, and that $\pi = (\pi_1,\dots, \pi_n)$ implements $\mu = (\mu_1, \dots, \mu_n)$ if $\pi_i$ implements $\mu_i$ for all $i$. The set of implementable marginals for player $i$ is denoted $M_i \subseteq \Delta(A_i)$.
\end{definition}


Given an implementable marginal $\mu$, we abuse notation to define the expected payoff to player $i$ following $\mu$ as
$U_i(\mu) = \mathbb{E}_{S \sim \rho}\left[\mathbb{E}_{a \sim \mu}[u_i(a)]\right]$. Note that by definition of $\mu$, $\mathbb{E}_{S \sim \rho}\left[\mathbb{E}_{a \sim \pi(S)}[u_i(a)]\right] = \mathbb{E}_{S \sim \rho}\left[\mathbb{E}_{a \sim \mu}[u_i(a)]\right]$. Similarly, we let $U_i(a_i; \mu_{-i}) = \mathbb{E}_{S \sim \rho}\left[\mathbb{E}_{a_{-i} \sim \mu_{-i}}[u_i(a_i, a_{-i})]\right]$.
The original strategy space is the finite product \(\prod_{S_i\in\mathcal S_i}\Delta(S_i)\), which is compact and convex. Since the map \(\pi_i\mapsto\mathbb E_{S_i\sim\rho_i}[\pi_i(\cdot\mid S_i)]\) is linear and continuous, its image \(M_i\) is also compact and convex. Consequently, we can take each player's ``effective'' strategy space to be $M_i$, implying an equivalence between Nash equilibria in the sense of \cref{defn:ne} and implementable marginal distributions that disincentivize unilateral deviations, as stated below. 

\begin{restatable}{proposition}{neImplementable}
    \label{thm:ne-implementable}
    Consider a GSAS where $\pi = (\pi_1, \dots, \pi_n)$ is a strategy profile that implements $\mu = (\mu_1, \dots, \mu_n)$. Then $\pi$ is an $\epsilon$-Nash equilibrium if and only if for all $i \in [n]$, 
        $U_i(\mu) \geq \max_{\mu_i' \in M_i} U_i(\mu_i', \mu_{-i}) - \epsilon$.
\end{restatable}
Notice that multiple $\pi$ may implement the same $\mu$, but not vice versa. \cref{thm:ne-implementable} implies that if $\pi$ and $\pi'$ both implement $\mu$ \emph{and} $\pi$ is a NE, then so is $\pi'$. In 2p0s-GSAS, a NE corresponds to a bilinear saddle point problem over $M_1$ and $M_2$, and a minimax theorem over implementable strategies holds.
\begin{restatable}
[Minimax theorem for 2p0s-GSAS]
{proposition}{neZeroSum}
    \label{thm:ne-2p0s}
    For a 2p0s-GSAS, a strategy profile $(\pi_1^*, \pi^*_2)$ is a NE if and only if the associated $(\mu^*_1, \mu^*_2)$ is a saddle point of the function $U=U_1=-U_2$, i.e.,
    $
        \mu^*_1 \in \argmax_{\mu_1 \in M_1}$ $ \min_{\mu_2 \in M_2} U(\mu_1, \mu_2)$ and
         $\mu^*_2 \in$ $ \argmin_{\mu_2 \in M_2}\max_{\mu_1 \in M_1} U(\mu_1, \mu_2)$,
        where
        $\max_{\mu_1 \in M_1}$ $ \min_{\mu_2 \in M_2} U(\mu_1, \mu_2) = \min_{\mu_2 \in M_2}$ $\max_{\mu_1 \in M_1} U(\mu_1, \mu_2)$.
\end{restatable}

The following shows that NE in $\mathcal{G}_\mathrm{orig}$ correspond to NE in $\mathcal{G}$, if the original NE are implementable. 
\begin{restatable}{proposition}{gsasContainsNe}
Consider a GSAS $\mathcal{G} = (\mathcal{G}_\mathrm{orig}, \mathcal{S}, \rho)$. Let $x^*=(x^*_1,\dots,x^*_n)$ be an $\epsilon$-NE of $\mathcal{G}_\mathrm{orig}$, $x^*_i \in \Delta(A_i)$. 
 If $\mu^*=x^*$ is implementable in $\mathcal{G}$ by $\pi^*=(\pi^*_1, \dots, \pi^*_n)$, then $\pi^*$ is an $\epsilon$-NE in~$\mathcal{G}$. 
\label{thm:gsas-contains-ne}
\end{restatable}
Note that even in the two-player zero-sum case, \cref{thm:gsas-contains-ne} requires the NE in $\mathcal{G}_\mathrm{orig}$, $x^*$, to be implementable for both players, i.e., $x^*_i$ alone being implementable does not imply a solution to the maximin problem or optimal strategy for player $i$ (see \cref{ex:matching_pennies} (ii)).
The above discussion indicates that for the purposes of equilibrium representation, we can work with $M_i$ rather than the larger set of possible $\pi_i$. However, it is still unclear how to recover $\pi_i$ from $\mu_i$ efficiently. 
We next show
that every $\mu_i \in M_i$ can be \emph{approximately} implemented
by a compact, polynomially sized weight vector $w_i$. 

\begin{restatable}{theorem}{gsasCompact}
    Let $\pi_i$ implement $\mu_i \in M_i$.
    Then, \(\forall \delta \in (0,1)\),
    there exists $w_i \in \Delta(A_i)$ prescribing a strategy \(\pi_i'\) where
		$\pi'_i(a_i \mid S_i) 
		:= \frac{w_i(a_i)\bbone\{a_i \in S_i\}}
		{\sum_{a'_i \in A_i} w_i(a'_i)\bbone\{a'_i \in S_i\}}$.
 Furthermore, \(\pi_i'\) induces a marginal \(\mu'_i\) with total variation distance \(d_{\mathrm{TV}}(\mu_i, \mu'_i) < \delta\).
	\label{thm:compact}
\end{restatable}
Intuitively, any implementable $\mu_i$ can be $\delta$-approximated in terms of total variation distance using a compact representation by an $|A_i|$-dimensional vector $w_i$.  
Given $w_i$, the corresponding $\pi'_i(a_i\mid S_i)$ for a fixed $S_i$ can be computed in time linear in $|A_i|$, playing proportionally to $w_i$ but restricted to the available actions $S_i$. 
We show an example of this representation in~\cref{ex:rps}.

\begin{example}
    \label{ex:rps}
    Consider a 2p0s-GSAS $\mathcal{G}$ where $\mathcal{G}_\mathrm{orig}$ is a variant of the standard rock-paper-scissors (RPS) game, where the payoff is halved whenever player 1 plays paper, i.e., the payoff is $0.5$ (resp. $-0.5$) instead of $1$ (resp. $-1$) for winning (resp. losing). The full payoff matrix of $\calG_{\mathrm{orig}}$ is given in~\cref{def:rps-mod}. For simplicity, denote the actions \textbf{R}ock, \textbf{P}aper, \textbf{S}cissors as \textbf{R}, \textbf{P}, \textbf{S}, respectively.
    

Let $\mathcal{S}_1=\{ S_{11}=\{ \textbf{R}, \textbf{P}, \textbf{S} \}, S_{12} = \{ \textbf{P}, \textbf{S} \} \}$ and $\mathcal{S}_2 = \{ S_{21} = \{ \textbf{R}, \textbf{S} \} , S_{22} = \{ \textbf{P}, \textbf{S}\} \}$ and $\rho_1(S_{11}) = \rho_1(S_{12}) = \rho_2(S_{21}) = \rho_2(S_{22}) =0.5$. It is easy to verify that a possible NE $\pi^*$ is $\pi^*_1(S_{11}) = \left(\textbf{R}: \frac{1}{2}, \textbf{P}: 0, \textbf{S}: \frac{1}{2}\right)$, $\pi^*_1(S_{12}) = \left(\textbf{R}: 0, \textbf{P}: 1, \textbf{S}: 0\right)$, $\pi^*_2(S_{21}) = \left(\textbf{R}: \frac{2}{3}, \textbf{P}: 0, \textbf{S}: \frac{1}{3}\right)$, and $\pi^*_2(S_{22}) = \left(\textbf{R}: 0, \textbf{P}: \frac{2}{3}, \textbf{S}: \frac{1}{3}\right)$. Both players obtain an expected payoff of $0$.
The corresponding marginal distributions of play are $\mu^*_1 = \left(\textbf{R}: \frac{1}{4}, \textbf{P}: \frac{1}{2}, \textbf{S}: \frac{1}{4}\right)$ and $\mu^*_2 = \left(\textbf{R}: \frac{1}{3}, \textbf{P}:\frac{1}{3}, \textbf{S}: \frac{1}{3}\right)$. 
It can be shown that compact vectors $w^*_1 = \left(\textbf{R}:\frac{1}{2}, \textbf{P}: \frac{1}{3}, \textbf{S}: \frac{1}{6}\right)$ and $w^*_2 = \left(\textbf{R}: \frac{2}{5}, \textbf{P}: \frac{2}{5}, \textbf{S}: \frac{1}{5}\right)$ are consistent with the marginal distribution and yield a NE of $\calG$ (see~\cref{appsec:analyticalRPS} for a derivation).
\end{example}

Utilizing a compact vector to recover appropriate choices from a stochastic set of actions has been studied in the literature of the Luce model of stochastic choice (see e.g.,~\cite{azrieli2026marginal}). Nevertheless, our techniques further allow us to establish novel approximate bit-complexity results for the GSAS setting (cf.~\cref{rem:bitcompact}).
Moreover, the compactness result can also be viewed as an analogue of Kuhn's theorem~\citep{kuhn1953extensive} for extensive-form games, which shows that behavioral strategies are \emph{outcome-equivalent} to the larger set of normal-form strategies under the assumption of perfect recall. Similarly,~\cref{thm:compact} ensures that compact vectors suffice to \emph{approximately} capture implementable NE under~\cref{assumption:product}. 

\begin{remark}\label{rem:bitcompact}
    Thus far, we have used the term ``compact representation'' rather loosely to refer to a simplex vector of size $\vert A_i\vert$. The astute reader might be concerned about their precise \emph{bit complexity}, i.e., whether the entries of $w_i^*$ can be written using a polynomial number of bits with respect to $\calG_{\mathrm{orig}}$. In~\cref{appsec:bitcomplexity}, we derive several results in this direction. First, we formalize the notion of a ``bit-compact'' representation in the context of GSAS strategies.
We show in~\cref{thm:approxbitcompactexist} and Corollary~\ref{cor:approxbitcompact} that under~\cref{assumption:product}, 2p0s-GSAS admit \emph{approximate} bit-compact representations. Conversely, 
we
show in~\cref{thm:nocompact} that without~\cref{assumption:product}, there exist families of 2p0s-GSAS that do \emph{not} admit bit-compact representations. 
This further justifies the importance of~\cref{assumption:product} in our analysis, and motivates future work on finding subclasses of GSAS that always admit bit-compact representations, even after relaxing the assumption.
\end{remark}

\section{Computing Equilibria in GSAS}\label{sec:compute}

In this section, we focus on the problem of NE computation in 2p0s-GSAS. There are three main regimes of interest. (i) The \textbf{small-support} regime, where $\mathcal{S}$ is small and $\rho$ is known exactly (cf. \cref{ex:rps}). (ii) The \textbf{oracle-access} regime, where $\mathcal{S}$ is exponentially large in $|A_i|$ so that $\rho$ cannot be explicitly enumerated, but can be queried in constant time for every $S \in \mathcal{S}$.
(iii) The \textbf{sample-access} regime, where $\rho$ is unknown and we only have access to a simulator that samples $S \sim \rho$. 
In regime (i), one could represent the game as a Bayesian game and apply off-the-shelf Bayesian game solvers. However, this does not take advantage of the additional structure afforded by GSAS, and incurs runtime costs linear in the number of types, which can be exponential in $\vert A_i\vert$ (cf. \cref{appsec:bayesiangame}). Our goal is to design a broad approach that applies to any GSAS, even in regime (iii).
All proofs of the theoretical results in this section are deferred to~\cref{appsecs:proofs}.

\subsection{Ranking regret minimization in GSAS}\label{sec:rankingregmin}
Going forward, we utilize the \emph{online learning} or \emph{repeated} game paradigm. In this setting, for each player $i\in [n]$, the sequence $\{S^t_i\}_{t=1, \ldots, T}$ is sampled i.i.d. from $\rho_i$. In each iteration $t$, player $i$ selects a strategy $\pi_i^t$, observes $S_i^t$, samples $a_i^t\sim\pi_i^t(\cdot\mid S_i^t)$, and then observes $u_i(a_i,a_{-i}^t)$ for every $a_i\in S_i^t$. Notice that this setting applies to all regimes (i)-(iii) as described above. A standard performance metric for learning in games is (cumulative) \emph{external regret}, formally defined as $\max_{a'_i\in A_i}\sum_{t=1}^T \left[u_i(a'_i,a^t_{-i}) - u_i(a^t_i,a^t_{-i})\right]$, with the folk result that no-external-regret algorithms 
converge in time-average to the set of Nash equilibria in 2p0s games. 
In GSAS, standard notions of regret are unsuitable since the competing action may be unavailable in certain rounds. This motivates the adoption of a regret variant called \emph{ranking regret} that was first studied in the sleeping bandit literature~\citep{kleinberg2010regret, kanade2009sleeping}.


Some additional notation is in order: a ranking $\tau_i$ is a total order on $A_i$. The corresponding \emph{ranking strategy} $\pi_i^{\tau_i}$ selects $\operatorname{top}_{\tau_i}(S_i)$, the highest-ranked action in the realized availability set $S_i$. Let $\mathcal T_i$ denote the set of rankings of $A_i$, and let $\mu_i^{\tau_i}\in M_i$ be the marginal induced by $\pi_i^{\tau_i}$. 
In the GSAS setting, the regret definition is given by comparing the players' \emph{expected} payoffs $U_i$ for playing strategy $\pi^t_i$ with the expected payoff for playing according to the ranking strategy $\pi^{\tau_i}_i$, which is selected in hindsight. 

\begin{definition}[Ranking Regret]
\label{def:ranking-regret}
Consider a sequence of strategy profiles $\pi^t=(\pi_1^t,\pi_2^t,\dots,\pi_n^t)$, $t=1,\ldots,T$. Player $i$'s \emph{ranking regret} is
    $R_{T,i}^{\mathrm{rank}}
    :=
    \max_{\tau_i\in\mathcal T_i}
    \sum_{t=1}^T
    \left[
        U_i(\pi_i^{\tau_i},\pi_{-i}^t)
        -
        U_i(\pi^t)
    \right]$.
\end{definition}


If a player's ranking regret grows sublinearly, i.e., $R_{T,i}^{\mathrm{rank}}=o(T)$,
they are said to have \emph{no-ranking-regret}. Intuitively, the player does not regret playing their strategy $\pi_i^t$ against the best ranking over actions in expectation.
Next, while ranking regret is defined for arbitrary GSAS, in~\cref{prop:SINash} we establish a connection between algorithms that achieve no-ranking-regret and Nash equilibria in the special case of 2p0s-GSAS. Let $\mu_i^t$ be the marginal distribution induced by $\pi_i^t$, and define the time-averaged induced marginal distribution
    $\bar\mu_i^T
    :=
    \frac1T\sum_{t=1}^T\mu_i^t$.
This marginal is implementable by definition, since the pointwise average strategy
    $\bar\pi_i^T(\cdot\mid S_i)
    :=
    \frac1T\sum_{t=1}^T\pi_i^t(\cdot\mid S_i)$
induces $\bar\mu_i^T$.
In 2p0s-GSAS, let $U \coloneqq U_1 = -U_2$. We consider the standard NE convergence metric of \emph{saddle point residual} (SPR).
By \cref{thm:ne-implementable}, it suffices to consider the implementable strategies $\mu_i\in M_i$ for each player $i$.

\begin{definition}[SPR in implementable strategies]
\label{def:SPR} The saddle point residual (\emph{SPR}) in implementable strategies is defined as
    \begin{align}
    \operatorname{SPR}(\mu_1,\mu_2)
    :=
    \max_{\mu_1'\in M_1}U(\mu_1',\mu_2)
    -
    \min_{\mu_2'\in M_2}U(\mu_1,\mu_2').
    \label{eq:SPR}
\end{align}
\end{definition}

\begin{restatable}{proposition}{SINashConvergence}
    \label{prop:SINash}
Consider a 2p0s-GSAS $\calG$ satisfying \cref{assumption:product}. For every sequence of strategies,
$
    \operatorname{SPR}(\bar\mu_1^T,\bar\mu_2^T)
    =
    \frac{
        R_{T,1}^{\mathrm{rank}}
        +
        R_{T,2}^{\mathrm{rank}}
    }{T}
    =:
    \epsilon_T.$
Consequently, every strategy profile implementing $(\bar\mu_1^T,\bar\mu_2^T)$ is an $\epsilon_T$-Nash equilibrium. In particular, if $R_{T,i}^{\mathrm{rank}}=o(T)$ for both players, then the time-averaged marginals converge to the set of NE of $\calG$ in terms of SPR.
\end{restatable}
In practice, we use the empirical marginals $\tilde{\mu}^T_i := \frac{1}{T}\sum_{t=1}^T \pi^t_i(\cdot \mid S^t_i)$ in place of $\bar{\mu}^T_i$; a standard concentration argument shows that the same SPR guarantee, up to an additional error term, holds with high probability.

\begin{remark}[Necessity of Ranking Regret]\label{rem:externalfails}
    While sublinear external regret suffices to show time-averaged convergence to the set of NE in standard 2p0s games, this implication does not hold in GSAS. In the sleeping bandits literature, other natural analogues of regret compatible with action restrictions are sleeping {external} and {internal} regret. However, we show in~\cref{appsec:counter} that neither notion is sufficient to guarantee convergence to Nash equilibrium in GSAS. Due to this, many prior algorithms which were designed for the sleeping bandit setting 
    do not apply to {equilibrium computation} in our setting.
\end{remark}
Proposition \ref{prop:SINash} reduces equilibrium computation to the problem of achieving sublinear \emph{expected pathwise} ranking regret. Here ``pathwise'' means that the best fixed ranking is chosen
after the randomly sampled trajectory is realized. We address this problem in the \emph{sample-access} setting (regime (iii)), where we assume that player $i$ can draw independent samples from $\rho_i$ in addition to observing the realized set $S_i^t$. 

Our proposed algorithm, CAT-FTPL, uses some ideas from prior work. 
Follow-the-Perturbed-Leader (FTPL)~\citep{hannan1957approximation} selects an available action by perturbing cumulative estimated losses; \textsc{SleepingCat}~\citep{neu2014online} uses Counting Asleep Times (CAT) to estimate the inverse of action-availability probabilities; and Geometric Resampling (GR)~\citep{neu2016importance} provides a way to estimate these probabilities by repeatedly sampling from $\rho$ (which is possible since we are in regime (iii)). CAT-FTPL combines FTPL with a GR implementation of a CAT-style estimator, together with a small implicit-exploration offset~\citep{kocak2014efficient} which yields the exponential-moment control needed to obtain sublinear regret. Full details of CAT-FTPL are given in~\cref{app:cat-ftpl} as~\cref{alg:cat-ftpl}.

\begin{restatable}{theorem}{CATFTPLRegret}
\label{thm:cat-ftpl}
A player using CAT-FTPL has expected pathwise ranking regret bounded as
\begin{equation}
    \mathbb E\left[R_{T,i}^{\mathrm{rank}}\right]
    \le 2\left(\log|A_i|\left(\frac{1}{\eta_i}
    +\frac{1}{\gamma_i}\right)
    +|A_i|T(\eta_i+\gamma_i)\right).
    \label{eq:cat-ftpl-bound}
\end{equation}

In particular, with parameters
    $\eta_i=\gamma_i=
    \sqrt{\frac{\log|A_i|}{|A_i|T}}$,
we obtain $\mathbb E\left[R_{T,i}^{\mathrm{rank}}\right]\le 8\sqrt{|A_i|T\log|A_i|}.$
\end{restatable}
Going forward, unless otherwise stated, we utilize CAT-FTPL with $\eta_i=\gamma_i=
    \sqrt{{\log|A_i|}/{|A_i|T}}$.
We additionally derive a high-probability bound beyond expected ranking regret, establishing that every player using CAT-FTPL has sublinear ranking regret with high probability.
\begin{restatable}{proposition}{HighProbRegret}
\label{prop:cat-ftpl-high-prob-regret}
Suppose players $i\in[n]$ run CAT-FTPL in a GSAS. 
Then for any $p_1,\ldots,p_n\in(0,1)$ with $\sum_ip_i<1$,
$\Pr\left(\bigcap^n_{i=1}\left\{
    R_{T,i}^{\mathrm{rank}}
    \le \frac{16}{p_i}
       \sqrt{|A_i|T\log |A_i|}
       \right\}
\right)
\ge 1-\sum_{i=1}^n p_i$.
\end{restatable}

Consequently, if the GSAS is 2p0s, Proposition \ref{prop:SINash} directly implies an
approximate-NE guarantee for the time-averaged induced marginals produced by CAT-FTPL.

\subsection{Computing compact equilibria in GSAS}
While the sublinear ranking regret guarantee of CAT-FTPL ensures that the time-averaged marginals converge to the set of NE in probability, these iterates alone do not constitute an explicit solution to the GSAS. In light of~\cref{thm:compact}, we seek a method to compute approximate \emph{compact} vectors associated with these averaged strategies.
For the following asymptotic result, suppose that a sequence of strategies \(\{\pi_i^t\}_{t=1,\dots,T}\) satisfies \(\frac1T\sum_{t=1}^{T}\pi_i^t(S_i^t)\to\mu_i^*\) almost surely as \(T\to\infty\), where \(\mu_i^*\) is a marginal distribution induced by some (approximate) Nash equilibrium \(\pi^*\).
Specifically, for a vector $w_i\in\mathbb{R}^{|A_i|}_{> 0}$, let $\hat\pi_i(a_i\mid S_i, w_i)$ be the probability of playing action $a_i$ given availability set $S_i$. Then, by~\cref{thm:compact} there exists a $\delta$-approximate $w_i$ so that $\forall S_i\in\mathcal{S}_i,\ a_i\in A_i$, $\hat\pi_i(a_i\mid S_i, w_i) = \frac{w_i(a_i)\bbone\{a_i\in S_i\}}{\sum_{a'_i\in A_i}w_i(a'_i)\bbone\{a'_i\in S_i\}}$.
Let $\hat \mu_i(w_i)$ be the corresponding marginal distribution where $\hat\mu_i(a_i|w_i) = \mathbb{E}_{S_i\sim \rho_i}[\hat\pi_i(a_i\mid S_i, w_i)]$ for all $a_i\in A_i$. Our objective is to find a $w_i$ that approximately solves $\hat \mu_i(w_i) = \mu_i^*$. 
This can be done through a stochastic approximation (SA) procedure as outlined in~\cref{appsec:5.7} as~\cref{alg:compute_w}. 

\begin{restatable}{proposition}{stochasticapproximation}
    \label{prop:compute_w} Let $w^T_i$ be the weight vector produced by Algorithm \ref{alg:compute_w}.
    If $\frac{1}{T}\sum_{t=1}^T\pi_i^t(S^t_i)\to \mu_i^*$ almost surely as $T\to \infty$, and that $\sum_{t=1}^{\infty}\eta_t = \infty$ and $\sum_{t=1}^{\infty}\eta^2_t < \infty$, then the marginal induced by $w^T_i$ converges to $\mu^*_i$ almost surely. In particular, almost surely, for every $\delta > 0$, $w^T_i$ is a $\delta$-approximate compact representation of $\mu^*_i$ for all sufficiently large \(T\).
\end{restatable}



Indeed, in our setting,~\cref{alg:compute_w} uses the time-averaged, empirical marginal distributions \(\tilde{\mu}_i^t:=\frac1t\sum_{\tau=1}^{t}\pi_i^\tau(S_i^\tau)\), where \(\pi_i^\tau\) are the strategies generated by CAT-FTPL. 
In addition to the result above, we also utilize a modified version of  the \emph{robust stochastic approximation} (RSA) approach introduced by~\citet{nemirovski1978cezari} and further developed by \citet{nemirovski2009robust} to obtain finite-time convergence bounds by running~\cref{alg:compute_w_rsa} and taking robust time-averages over the iterates (see~\cref{appsec:robustSA} for details). These results imply a simple procedure for computing \emph{compact} approximate equilibria in GSAS: (i) run CAT-FTPL for rounds $t=1,\dots,T$, (ii) compute the empirical marginal distributions $\tilde\mu^t_i$ for each $t$, and (iii) run Algorithm~\ref{alg:compute_w_rsa} to obtain $\tilde w_i$.
In~\cref{appsec:exampleRSP}, we show an example of this procedure applied to~\cref{ex:rps}, approximately recovering $w_1$ and $w_2$ that implement an \(\epsilon\)-NE. Taken together, our analysis leads to a high-probability  \emph{finite-time} bound on the SPR of the computed approximate compact vectors.

\begin{restatable}{theorem}{mainthmnormalform}
\label{thm:mainthm}
    Suppose CAT-FTPL is run for $T$ rounds 
    in a 2p0s-GSAS with payoffs $u_i:A\to[-1,1]$, and the empirical marginals \((\tilde \mu_1, \tilde \mu_2)\) are used in~\cref{alg:compute_w_rsa} to obtain compact vectors $\tilde{w}_i$.
    Then, for any $p\in (0,1)$, with probability  at least $1-p$, the SPR of the strategies induced by $\tilde{w}_i$ satisfies $\mathrm{SPR}^2\leq O\left(\frac{1}{p\sqrt{T}} + \frac{1}{p^2T}
    \right)$, where $O(\cdot)$ hides polynomial factors in $|A_i|$.
\end{restatable}
CAT-FTPL contributes \(O(1/(p^2T))\) to the bound, which is dominated by the \(O(1/(p\sqrt T))\) term from the RSA procedure whenever \(\sqrt{T} \gg 1/p\).
Note that the bound on SPR is relatively loose, since the RSA analysis bounds the expected squared $\ell_2$-norm of the error induced by the computed compact vector, whereas SPR is controlled by the $\ell_1$-norm of this error. Nevertheless, asymptotic convergence in terms of SPR still holds in probability, and we leave as an open problem the design of SA procedures that improve upon this bound. Moreover, we show  in~\cref{sec:experiments} that our procedure empirically computes $\tilde{w}$ with low SPR in large 2p0s-GSAS.

\section{Experimental Results}\label{sec:experiments}
\paragraph{Experiment 1: Comparison with LP solver.}
In light of our compact representation results, we investigate the \emph{scalability} of CAT-FTPL compared to existing solvers for GSAS with exponentially sized support \(\calS\).
Since there are no existing baselines for GSAS solvers, we compare our method against reductions that allow for the application of linear programs (LPs). Note, however, that the LP method requires full access to \(\rho\), while CAT-FTPL only requires sample access to $\rho$. Due to the inherent differences between LP solvers and our online learning approach, the computed solutions are not directly comparable in terms of output quality, but only in \emph{runtime}.
In~\cref{fig:runtime_experiment}, we compare the wallclock runtime of CAT-FTPL and Gurobi's linear program solver on randomly generated GSAS (full definition in~\cref{def:randomgsas}). Fixing a time budget of 360 seconds, we found that Gurobi's LP solver, when applied to the sequence-form representation~\citep{von1996efficient} of the GSAS that encodes all action subsets, could only solve games of up to size \(11 \times 11\). Comparatively, CAT-FTPL was able to obtain low marginal SPR, i.e., \(\mathrm{SPR}(\tilde \mu_1, \tilde \mu_2) \leq 0.05\), for much larger games ($350\times 350$) within the same time budget, demonstrating the scalability of our approach.
Additional experimental details are given in \cref{appsubsec:exp1_additional_details}. Finally, beyond comparisons with LP solvers, we also compare our method with a reduction that allows for the application of Bayesian game solvers (details in~\cref{appsubsec:multiscalemwu}). 

\begin{figure}[ht]
    \centering
    \includegraphics[width=0.6\linewidth]{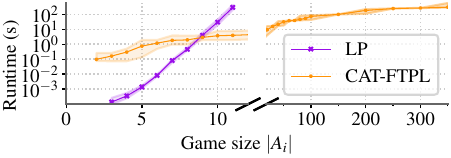}
    \caption{Wallclock time to solve randomly generated GSAS by CAT-FTPL and Gurobi LP solver. Plot shows the average over 20 runs with shaded region showing the range (max and min of wallclock time) of values across runs.}
    \label{fig:runtime_experiment}
\end{figure}

\begin{figure}[ht]
    \centering
    \begin{subfigure}[t]{0.48\linewidth}
        \vspace{0pt}
        \includegraphics[width=\linewidth]{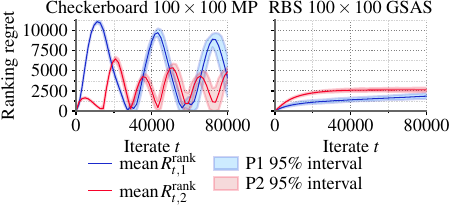}
            \caption{Ranking regret \(R^\mathrm{rank}_{t,i}\) for intermediate iterates with time horizon \(T=80000\). Mean and 95\% central interval from 100 runs shown.}
        \label{fig:regret_plots_iterates}
    \end{subfigure}
    \hfill
    \begin{subfigure}[t]{0.48\linewidth}
        \vspace{0pt}
        \includegraphics[width=\linewidth]{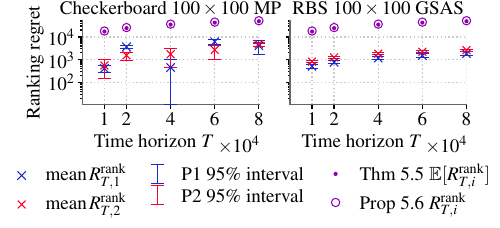}
        \caption{Ranking regret \(R^\mathrm{rank}_{T,i}\) of final iterate of CAT-FTPL for various time horizons \(T\). Results from 100 runs for each \(T\) and theoretical bounds shown.
        }
        \label{fig:regret_plots_final_iterate}
    \end{subfigure}
    \caption{Ranking regret of CAT-FTPL for two 2p0s-GSAS. Mean, central 95\% interval and theoretical expected/high-probability bounds on maximum ranking regret from~\cref{thm:cat-ftpl} and~\cref{prop:cat-ftpl-high-prob-regret} are shown. We observed \(R^\mathrm{rank}_{T,i}\) to be much lower in practice than the bound in \cref{thm:cat-ftpl}.}
    \label{fig:regret_plots}
\end{figure}

\begin{figure}[!ht]
    \centering
    \begin{subfigure}[t]{0.48\linewidth}
        \includegraphics[width=\linewidth]{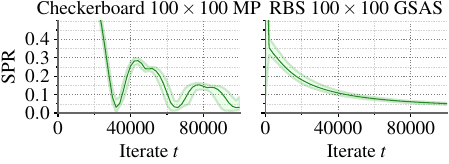}
        \caption{SPR of the time-averaged empirical marginals \(\tilde \mu_i^t\) played by CAT-FTPL at intermediate iterates for time horizon \(T=100000\). 
        }
        \label{fig:spr_plots_marginals}
    \end{subfigure}
    \hfill
    \begin{subfigure}[t]{0.48\linewidth}
        \centering
        \includegraphics[width=\linewidth]{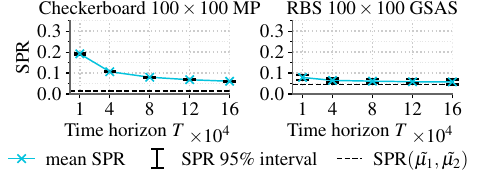}
        \caption{SPR of strategies induced by the $w_i^T$ vectors computed by \cref{alg:compute_w_rsa} for various time horizons \(T\).}
        \label{fig:spr_plots_w}
    \end{subfigure}
    \caption{SPR for two 2p0s-GSAS. For each game, we repeated each experiment for 100 runs and plotted the mean and 95\% central interval over the runs.}
    \label{fig:spr_plots}
\end{figure}

\paragraph{Experiment 2: Convergence in large GSAS.}

We seek to corroborate our convergence results in~\cref{sec:compute} in terms of ranking regret and SPR.
We focus on two large GSAS: (i) $100\times100$ Random Biased Support GSAS (``RBS'') (cf.~\cref{def:RBSgsas}) and (ii) Checkerboard \(100 \times 100\) Matching Pennies (cf.~\cref{def:checkermp}), which are specifically constructed to be challenging due to the payoff structure and stochastic action availabilities. 
Additional experimental setup details, including estimating \(R^{\mathrm{rank}}_{T,i}\) and SPR, as well as further experiments on other GSAS, are given in~\cref{appsec:experimentaldetails}.
In~\cref{fig:regret_plots} we focus on the ranking regret, with~\cref{fig:regret_plots_iterates} showing \(R^\mathrm{rank}_{t,i}\) at all intermediate iterates \(t\in\{1,...,T\}\) for a single time horizon \(T\).~\cref{fig:regret_plots_final_iterate} shows the final ranking regret at time horizon \(T\) for various values of \(T\), and we find the observed \(R^\mathrm{rank}_{T,i}\) to be bounded by~\cref{thm:cat-ftpl} and~\cref{prop:cat-ftpl-high-prob-regret}, consistent with sublinear growth with high probability.
In~\cref{fig:spr_plots} we focus on the SPR, with~\cref{fig:spr_plots_marginals} showing the time-average marginal \(\tilde \mu_i^t\) at the intermediate iterates \(t \in \{1,...,T\}\)  played by CAT-FTPL,
which we observe converges to low SPR marginals even in challenging games by time horizon \(T\). In~\cref{fig:spr_plots_w}, we show the SPR of the final compact strategy \(w_i^T\) computed by~\cref{alg:compute_w_rsa} for various time horizons \(T\), using as input marginals \((\tilde\mu_1, \tilde\mu_2)\) produced by CAT-FTPL. We empirically observe that the SPR corroborates~\cref{thm:mainthm},
indicating that our proposed method is effective in computing compact strategies representing \(\epsilon\)-NE in 2p0s-GSAS. 

\section{Discussion and Future Work}\label{sec:discussion}
In this paper we have taken the first step toward characterizing and computing \emph{compact} Nash equilibria in games with stochastic action sets. Our analysis leaves open several promising future research directions. These include (i) characterizing and studying convergence to appropriate notions of \emph{correlated} equilibria in general-sum GSAS, (ii) modeling and solving extensive-form and Markovian variations of GSAS, (iii) exploring efficient strategy representations even after relaxing the independence assumption, and (iv) designing more computationally efficient, parameter-free ranking-regret algorithms, particularly methods that avoid  repeated geometric resampling.

\newpage
\section*{Acknowledgements}
This project is supported by the Ministry of Education, Singapore, under the Academic Research Fund Tier 1 (FY2025) and by the National University of Singapore, under the Start-Up Grant Scheme. The authors thank Le Van Cuong for his contributions during the initial stage of the project.

\section*{Impact Statement}
This paper presents work whose goal is to advance the field of machine
learning and game theory. There are many potential societal consequences of our work, none of
which we feel must be specifically highlighted here.

\bibliographystyle{plainnat}
\bibliography{ref}

\newpage
\appendix


\section*{Appendix}
This supplementary material contains an overview of additional related work in Appendix~\ref{appsecs:related}, proofs omitted from the main paper for space considerations in Appendix~\ref{appsecs:proofs}, 
further experimental results and details in Appendix~\ref{appsecs:experiments}, and formal bit-complexity results in~\cref{appsec:bitcomplexity}.

\section{Additional Related Work}\label{appsecs:related}

\paragraph{Games with Action Set Restrictions.} \citet{benaim2010class} and \citet{bravo2015reinforcement} studied an algorithm called Markovian Fictitious Play (MFP) in repeated two-player normal-form games where the action sets are restricted at each round. Unlike our setting, the action restrictions depend on the players' previous actions and are encoded via an \emph{exploration matrix}. MFP requires players to compute a best response at each round and is shown to converge almost surely to NE in two-player zero-sum and potential games.

\paragraph{Learning in Games with Exponential Action Sets.} 
\citet{panageas2023semi} and 
\citet{dong2023taming} studied semi-bandit learning in congestion games which admit exponentially large action sets, leading to slow convergence of standard methods. While their approach also utilizes learning over a compact set of `facilities', their setting focuses on computing NE of the original congestion game (i.e., NE which are implementable in a sense that is introduced in~\cref{def:marginal}). Comparatively, in GSAS, the set of NE can change drastically, and we seek to compute the NE of the game induced by the stochastic action sets, not $\calG_{\mathrm{orig}}$ itself. Moreover, our emphasis is also on computing compact representations of NE, which allow players to efficiently play the game optimally.

The study of regret minimization in games with structured strategy sets beyond congestion games has an extensive literature, primarily thanks to the canonical connection between no-regret learning and game-theoretic equilibria~\citep{freund1999adaptive,cesa2006prediction}. \citet{farina2022kernelized} and \citet{kontogiannis2025efficient} studied regret minimization in \emph{polyhedral} games (i.e., games with combinatorial structure which might have exponential action sets), and~\citet{farina2020stochastic} studied regret minimization under stochasticity (i.e., Monte-Carlo sampling of regrets). These papers focus on deriving regret bounds, whereas we additionally study efficient methods for equilibrium representation.




\subsection{On Connections to Bayesian Games}
\label{appsec:bayesiangame}
Since a GSAS can be viewed as a variant of Bayesian games, it is natural to discuss existing methods for online learning in Bayesian games. There are two primary access models for $\rho$ where algorithms for Bayesian games can be applied to GSAS. 

First, we consider the case where players only have sample or oracle access to $\rho$. In this setting,~\citet{hartline2015no} studied no-regret learning in finite Bayesian games, requiring only sample access to $\rho$. However, there are two key distinctions: (i) the equilibrium concept of concern in~\cite{hartline2015no} is coarse correlated equilibria (CCE) in general-sum Bayesian games, and so they focus on no-\emph{external}-regret algorithms, and (ii) their convergence result is an almost sure convergence statement and does not admit an explicit finite convergence rate. In comparison, we obtain finite-time convergence guarantees via a more fine-grained analysis.

Similarly, recent methods can approximately solve for correlated equilibria in general-sum Bayesian games~\citep{dagan2024external,peng2024complexity} using only sample access to $\rho$, but requiring knowledge of $K$, the number of types. Their algorithm, called multi-scale MWU, can in principle be applied to the 2p0s-GSAS setting, but the runtime grows as $O(\vert A\vert K \log(\vert A\vert K))$ (note that the runtime per iteration depends linearly on $K$). Comparatively, we exploit the additional structure of GSAS to design a procedure that has an improved rate and runtime that does \emph{not} depend on the number of types. In~\cref{fig:multiscale-MWU}, we compare the wallclock runtimes of multi-scale MWU and our method in randomly generated GSAS, showing that our method can solve significantly larger games than multi-scale MWU in a given time budget.


The second setting is in the case that $\rho$ is explicitly known, which makes equilibrium computation easier, since GSAS can be explicitly written as a Bayesian game where action subsets are modeled as types. Then, one could expand the GSAS into an expanded normal-form game, similar to how one might expand a Bayesian game into agent-form/induced normal-form, and solve it directly in this space. To improve scalability, one can also avoid the full-blown normal-form LP by converting the game into sequence-form, which models Nature as a chance node in a game tree that determines the players’ types. This representation is well-known to be computationally more efficient than the full normal-form representation, and solving this can be done with an off-the-shelf game solver which obtains a Bayesian Nash equilibrium corresponding to a Nash equilibrium in the GSAS. However, these methods have a dependence on the size of the sequence-form representation, which again can be linear in the number of types. Conversely, our procedure is designed to exploit the additional structure afforded by GSAS to avoid this blow-up. This observation is corroborated in our experiments (cf.~\cref{fig:runtime_experiment}), where our method is compared explicitly to an LP solver as applied to the sequence-form of the game.

\section{Omitted Proofs}\label{appsecs:proofs}
\subsection{Derivation of Equation (\ref{eq:depend-only-marginals})}\label{appsec:eqn4deriv}
\[
\begin{aligned}
U_i(\pi)
&= \mathbb{E}_{S\sim\rho}
   \mathbb{E}_{a\sim\pi(S)}[u_i(a)]\\
&= \sum_{S\in\mathcal S}
   \left(\prod_{j\in[n]}\rho_j(S_j)\right)
   \sum_{a\in S}
   u_i(a)\prod_{j\in[n]}\pi_j(a_j\mid S_j)\\
&= \sum_{a\in A}u_i(a)
   \sum_{\substack{S\in\mathcal S\\a\in S}}
   \prod_{j\in[n]}\rho_j(S_j)\pi_j(a_j\mid S_j)\\
&= \sum_{a\in A}u_i(a)
   \sum_{S\in\mathcal S}
   \prod_{j\in[n]}
   \rho_j(S_j)\pi_j(a_j\mid S_j)
   \mathbf{1}\{a_j\in S_j\}\\
&= \sum_{a\in A}u_i(a)
   \prod_{j\in[n]}
   \sum_{S_j\in\mathcal S_j}
   \rho_j(S_j)\pi_j(a_j\mid S_j)
   \mathbf{1}\{a_j\in S_j\}\\
&= \sum_{a\in A}u_i(a)
   \prod_{j\in[n]}P[a_j;\rho_j,\pi_j].
\end{aligned}
\]

The second equality holds by expanding the expectations, using \cref{assumption:product} to factor $\rho$ and the definition of the joint strategy to factor the conditional action distribution. The third exchanges two finite sums, the fourth replaces the constraint $a\in S$ with action-availability indicators, the fifth factors the sum over the Cartesian product, and the final equality follows from the definition.

\subsection{Proof of Proposition~\ref{thm:ne-implementable}}
\neImplementable*
\begin{proof}
    ($\impliedby$) Consider a strategy profile $\pi$ which implements $\mu$.
    
    We have that $U_i(\mu) \geq \max_{\mu_i' \in M_i} U_i(\mu_i', \mu_{-i}) - \epsilon$.
    Expanding the expression for the expected payoff of $\pi$, we get
    \begin{align*}
        U_i(\pi) &= \sum_{a\in A} u_i(a) \prod_{j\in [n]} \mathbb{P}[a_j ; \rho_j, \pi_j]\\
        &= U_i(\mu)\\ 
        &\ge \max_{\mu_i' \in M_i} U_i(\mu_i', \mu_{-i}) - \epsilon\\
        &= \max_{\pi_i'} \sum_{a_i \in A_i} u_i(a) \mathbb{P}[a_i; \rho_i, \pi_i'] \prod_{j \neq i} \mathbb{P}[a_j; \rho_j, \pi_j] - \epsilon\\
        &= \max_{\pi_i'} U_i(\pi_i', \pi_{-i}) - \epsilon,
    \end{align*}
    where we utilize the implementability of $\mu$ and the fact that $\mu_i(a_i) = \mathbb{P}[a_i; \rho_i, \pi_i]$.

    ($\implies$)
    Conversely, suppose that \(\pi\) is an \(\epsilon\)-Nash
    equilibrium. Fix \(i\) and an arbitrary \(\mu_i'\in M_i\), and
    choose a strategy \(\pi_i'\) that implements \(\mu_i'\). Then
    \[
    U_i(\mu)
    =U_i(\pi)
    \ge U_i(\pi_i',\pi_{-i})-\epsilon
    =U_i(\mu_i',\mu_{-i})-\epsilon.
    \]

    Taking the maximum over \(\mu_i'\in M_i\) proves the claim.
\end{proof}
\subsection{Proof of Proposition~\ref{thm:ne-2p0s}}
\neZeroSum*
\begin{proof}
    The proof follows directly from von Neumann's minimax theorem~\citep{v1928theorie} and the fact that $U(\pi) = U(\mu)$ if $\pi$ implements $\mu$. 
\end{proof}
\subsection{Proof of Proposition~\ref{thm:gsas-contains-ne}}
\gsasContainsNe*
\begin{proof}
Let \(\pi^\ast\) implement \(\mu^\ast=x^\ast\). Any unilateral deviation \(\pi_i'\) in the GSAS induces an implementable marginal \(\mu_i'\in M_i\subseteq\Delta(A_i)\). Since \(x^\ast\) is an \(\epsilon\)-Nash equilibrium of \(G_{\mathrm{orig}}\),
\[
u_i(x_i^\ast,x_{-i}^\ast)
\ge
u_i(\mu_i',x_{-i}^\ast)-\epsilon.
\]

As \(\pi^\ast\) implements \(x^\ast\),
\(
U_i(\pi^\ast)=u_i(x_i^\ast,x_{-i}^\ast)\), and \(U_i(\pi_i',\pi_{-i}^\ast)
=u_i(\mu_i',x_{-i}^\ast) \).
Therefore,
$
U_i(\pi^\ast)
\ge
U_i(\pi_i',\pi_{-i}^\ast)-\epsilon$
for every player \(i\) and every unilateral deviation
\(\pi_i'\), and by definition \(\pi^\ast\) is an \(\epsilon\)-Nash
equilibrium of \(G\).
\end{proof}
    
\subsection{Proof of Theorem~\ref{thm:compact}}\label{proof:compact}
\gsasCompact*
\begin{proof}
    We start by assuming that every action is available with strictly positive probability, i.e., \(\forall i \in [n]\, \forall a_i \in A_i,\, \Pr_{S_i \sim \rho_i}[a_i \in S_i]>0\), and that every availability set is available with strictly positive probability, i.e., \(\forall i \in [n]\, \forall S_i \in \calS_i,\, \rho_i(S_i) > 0\). At the end of the proof, we will relax this assumption.
    
	The proof will require the following linear algebraic result.
	\begin{theorem}[\cite{menon1969spectrum}]
		\label{theorem:matrix_scaling}
		Let $X \in \mathbb{R}_{\ge 0}^{n \times m}$ be a nonnegative matrix and $r \in \mathbb{R}_{>0}^n$, $c \in \mathbb{R}_{>0}^m$. Then there exist $u \in \mathbb{R}_{>0}^n$ and $v \in \mathbb{R}_{>0}^m$ such that $P = \operatorname{diag}(u)X\operatorname{diag}(v)$ has row sums $r$ and column sums $c$ if and only if there exists a nonnegative matrix  $Y \in \mathbb{R}_{\geq0}^{n \times m}$ with row sums $r$, column sums $c$, and $\operatorname{supp}(Y) = \operatorname{supp}(X)$.
	\end{theorem}
	
	Consider any player $i\in [n]$ and, overloading notation for the purposes of this proof, fix some \(\epsilon \in (0, \frac{\delta}{\lvert A_i \rvert})\).
	We will use Theorem \ref{theorem:matrix_scaling} to show the existence of $w_i$.
    Going forward, we
    apply Theorem \ref{theorem:matrix_scaling} for availability sets that have strictly positive support on $\rho$.
    For each availability set $S_i\in\mathcal{S}_i$, let us define a conditional distribution $\tilde \pi^{(\epsilon)}_i(S_i)$ as
	\begin{equation}\label{eqn:conditional}
		\tilde \pi^{(\epsilon)}_i(a_i\mid S_i) := 
		\begin{cases}
			(1-\epsilon)\pi_i(a_i\mid S_i) + \frac{\epsilon}{|S_i|}, &\text{if } a_i \in \mathcal S_i,\\
			0, &\text{otherwise}.
		\end{cases}
	\end{equation}
    
	This perturbation term $\frac{\epsilon}{|S_i|}$ ensures that \(\tilde \pi_i^{(\epsilon)}\) has full support on every \(S_i\). Let $\tilde \mu_i^{(\epsilon)}$ be the corresponding marginal distribution over $A_i$, i.e., $\tilde \mu_i^{(\epsilon)}(a_i) := \sum_{S_i\in\mathcal{S}_i}\rho_i(S_i)\tilde \pi^{(\epsilon)}_i({a_i\mid S_i})$.
    
	Let us also define a matrix $X_i\in\{0,1\}^{|\calS_i|\times|A_i|}$ as the indicator of action availability, i.e., $X_i(S_i,a_i):=\bbone\{a_i\in S_i\}$, and a nonnegative matrix $Y_i^{(\epsilon)}\in\mathbb{R}^{|\calS_i|\times|A_i|}_{\geq 0}$ with
	\(Y^{(\epsilon)}_i({S_i,a_i}) := \rho_i(S_i)\tilde \pi_i^{(\epsilon)}(a_i\mid S_i)\).
     
	The row sums of \(Y^{(\epsilon)}_i\) are given by
	\begin{align}
		\sum_{a_i \in A_i} Y^{(\epsilon)}_i(S_i, a_i) &= \sum_{a_i \in A_i} \rho_i(S_i)\tilde \pi_i^{(\epsilon)}(a_i\mid S_i)\\
		&= \rho_i(S_i) \left[ \sum_{a_i \in A_i}\bbone\{a_i \in \calS_i\}\left((1-\epsilon)\pi_i(a_i\mid S_i) + \frac{\epsilon}{|S_i|} \right) \right]\\
		&= \rho_i(S_i) \left[ (1-\epsilon) \sum_{a_i \in S_i} [ \pi_i(a_i\mid S_i)] + \epsilon \left(\sum_{a_i \in S_i}\frac{1}{|S_i|} \right)\right]\\
	&= \rho_i(S_i).\quad\text{(as both inner sums equal 1)}
	\end{align}
	
	The column sums of \(Y^{(\epsilon)}_i\) are given by 
	\(\sum_{S_i \in \calS_i} Y^{(\epsilon)}_i(S_i, a_i) = \sum_{S_i \in \calS_i} \rho_i(S_i)\tilde \pi_i^{(\epsilon)}(a_i\mid S_i) =: \tilde \mu^{(\epsilon)}_i(a_i) \),
	which is the marginal distribution of \(\tilde \pi^{(\epsilon)}_i\).
	
	Finally, note that \(Y_i^{(\epsilon)}(S_i, a_i) \neq 0\) if and only if \(a_i \in S_i\). This follows from the definition of \(\tilde \pi^{(\epsilon)}_i(a_i\mid S_i)\) (\cref{eqn:conditional}). Thus, \(\supp(Y_i^{(\epsilon)}) = \{(S_i,a_i) \mid a_i \in S_i\} = \supp(X_i)\).
	Since the matrix \(Y_i^{(\epsilon)}\) has row sums \(\rho_i(S_i)\), column sums \(\tilde \mu_i^{(\epsilon)}\) and \(\supp(Y_i^{(\epsilon)}) = \supp(X_i)\), the hypothesis of~\cref{theorem:matrix_scaling} holds for the triple $(X_i,\rho_i,\tilde\mu_i^{(\epsilon)})$. Thus, there exist vectors $u_i^{(\epsilon)}\in\mathbb R_{>0}^{|\mathcal{S}_i|}$ and $v_i^{(\epsilon)}\in\mathbb R_{>0}^{|A_i|}$ such that
	\(P_i^{(\epsilon)} = \operatorname{diag}\big(u_i^{(\epsilon)}\big) X_i \operatorname{diag}\big(v_i^{(\epsilon)}\big)\)
	has row sums $\rho_i$ and column sums $\tilde \mu^{(\epsilon)}_i$.
	 
	
	Now, we can define a probability distribution \(\pi'_i(S_i)\) over \(A_i\) for every availability set \(S_i\) as
	\begin{align}
		\pi'_i(a_i \mid S_i) &:= \frac{P^{(\epsilon)}_i(S_i,a_i)}{\sum_{a'_i \in A_i} P^{(\epsilon)}_i(S_i, a'_i)} \label{eqn:piprimedef}\\
		&= \frac{u_i^{(\epsilon)}(S_i) X_i(S_i, a_i) v_i^{(\epsilon)}(a_i)}{\sum_{a'_i \in A_i} u_i^{(\epsilon)}(S_i) X_i(S_i, a'_i) v_i^{(\epsilon)}(a'_i)}\\
		&= \frac{X_i(S_i, a_i) v_i^{(\epsilon)}(a_i)}{\sum_{a'_i \in A_i} X_i(S_i, a'_i) v_i^{(\epsilon)}(a'_i)}\\
		&=\frac{v_i^{(\epsilon)}(a_i)\bbone\{a_i\in S_i\}}{\sum_{a'_i\in A_i} v_i^{(\epsilon)}(a'_i)\bbone\{a'_i\in S_i\}}\\
		&=\frac{\frac{v_i^{(\epsilon)}(a_i)}{\| v_i^{(\epsilon)} \|_1}\bbone\{a_i\in S_i\}}{\sum_{a'_i\in A_i} \frac{v_i^{(\epsilon)}(a'_i)}{\| v_i^{(\epsilon)} \|_1}\bbone\{a'_i\in S_i\}}\\
		&=\frac{w_i(a_i)\bbone\{a_i\in S_i\}}{\sum_{a'_i\in A_i} w_i(a'_i)\bbone\{a'_i\in S_i\}},
	\end{align}    
	where we define \(w_i(a_i) := \frac{v_i^{(\epsilon)}(a_i)}{\| v_i^{(\epsilon)} \|_1}\) and note that due to the rescaling, \(w_i \in \Delta(A_i)\).
	Next, the marginal probability distribution \(\mu'_i \in M_i\) induced by \(\pi'_i\) is
	\begin{align}
		\mu'_i(a_i) &= \sum_{S_i \in \calS_i} \rho_i(S_i) \pi'_i(a_i \mid S_i)\\
	&= \sum_{S_i \in \calS_i} \rho_i(S_i) \frac{P_i^{(\epsilon)}(S_i,a_i)}{\sum_{a'_i \in A_i} P_i^{(\epsilon)}(S_i, a'_i)},
	\end{align}
    where the latter equality follows by~\cref{eqn:piprimedef}.
	As the row sums of \(P^{(\epsilon)}_i\) are \(\rho_i\), we have \(\sum_{a'_i \in A_i} P^{(\epsilon)}_i(S_i, a'_i) = \rho_i(S_i)\) and so
	\[
		\mu'_i(a_i) = \sum_{S_i \in \calS_i} P^{(\epsilon)}_i(S_i,a_i).
	\]

	By~\cref{theorem:matrix_scaling}, \(P^{(\epsilon)}_i\) and \(Y_i^{(\epsilon)}\) have equal column sums and thus
		\begin{align}
		\mu'_i(a_i) &= \sum_{S_i \in \calS_i} Y^{(\epsilon)}_i(S_i,a_i)\\
		 &= \sum_{S_i \in \calS_i} \rho_i(S_i) \tilde \pi^{(\epsilon)}_i(a_i\mid S_i)\label{129}\\
		&= \sum_{S_i \in \calS_i} \rho_i(S_i) \bbone\{a_i \in S_i\}\left[(1-\epsilon)\pi_i(a_i\mid S_i) + \frac{\epsilon}{|S_i|}\right],
	\end{align}
    where Equation \eqref{129} follows from the definition of \(Y_i^{(\epsilon)}\).
	Rearranging and noting \(\pi_i(a_i \mid S_i) = 0\) when \(a_i \notin S_i\),
	\begin{align}
		\mu'_i(a_i) 	&= (1-\epsilon) \sum_{S_i \in \calS_i} \rho_i(S_i) \pi_i(a_i \mid S_i) + \sum_{S_i \in \calS_i} \frac{\epsilon \rho_i(S_i) \bbone\{a_i \in S_i\}}{\lvert S_i \rvert}\\
			&= (1-\epsilon) \mu_i(a_i) + \sum_{S_i \in \calS_i} \frac{\epsilon \rho_i(S_i) \bbone\{a_i \in S_i\}}{\lvert S_i \rvert}\\
			&= \mu_i(a_i) + \epsilon \left(-\mu_i(a_i) + \sum_{S_i \in \calS_i} \frac{\rho_i(S_i) \bbone\{a_i \in S_i\}}{\lvert S_i \rvert}\right). \label{eqn:marginalrange}
	\end{align}

	We can now bound \(\mu'_i(a_i)\) on both sides. Using  \cref{eqn:marginalrange} and since \(\mu_i(a_i) \in [0,1]\),
	\begin{align}
		\mu'_i(a_i) &\geq \mu_i(a_i) - \epsilon \mu_i(a_i)\\
		 &\geq \mu_i(a_i) - \epsilon.
	\end{align}

	Similarly, from \cref{eqn:marginalrange} and utilizing the fact that \(\lvert S_i \rvert \geq 1\),	
	\begin{align}
		\mu'_i(a_i) & \leq \mu_i(a_i) + \epsilon \sum_{S_i \in \calS_i} \frac{\rho_i(S_i) \bbone\{a_i \in S_i\}}{\lvert S_i \rvert}\\
	&\leq \mu_i(a_i) + \epsilon \sum_{S_i \in \calS_i} \rho_i(S_i)\\
&\leq \mu_i(a_i) + \epsilon.
	\end{align}

	Thus, for all $a_i\in A_i$, \(\lvert \mu'_i(a_i) - \mu_i(a_i) \rvert \leq \epsilon\).
    
    The total variation distance between \(\mu_i\) and \(\mu'_i\) can then be bounded by
    \begin{align}
        d_{\mathrm{TV}}(\mu_i, \mu'_i) &= \frac{1}{2} \sum_{a_i \in A_i} \lvert \mu_i(a_i) - \mu'_i(a_i) \rvert\\
        &\leq \frac{1}{2} \sum_{a_i \in A_i} \epsilon\\
        &\leq \frac{1}{2} \lvert A_i \rvert \frac{\delta}{\lvert A_i \rvert}\label{eq:41}\\
        &< \delta\label{eq:42},
    \end{align}
    where~\cref{eq:41} uses $\epsilon<\delta/|A_i|$ and~\cref{eq:42} uses $\delta>0$.
    
    It now remains to relax our assumption requiring all actions to be available with nonzero probability, and all availability sets to be observed with nonzero probability.
    Indeed, any GSAS \(\calG\) not meeting this assumption can be reduced to an equivalent GSAS \(\calG'\) with actions and availability sets that are never available removed. Obviously, this process does not affect the game structure and the player strategies. The compact vector \(w_i'\) in the reduced GSAS \(\calG'\) can then be converted to an equivalent compact vector \(w_i\) in the original GSAS \(\calG\) by setting the entries for never-available actions to zero, and leaving entries for observed actions unchanged. This completes the proof.
\end{proof}

    \subsection{Illustration of proof via Example~\ref{ex:rps}}\label{appsec:analyticalRPS}
    We illustrate the proof of~\cref{thm:compact} using  \cref{ex:rps} from the main text.
    Utilizing ideas from the proof of~\cref{thm:compact} and the statement of Theorem~\ref{theorem:matrix_scaling}, we define a matrix $X_i\in\{0,1\}^{|\mathcal{S}_i|\times|A_i|}$ as the indicator of action availability, i.e., $X_i(S_i,a_i):=\bbone\{a_i\in S_i\}$, and a nonnegative matrix $Y_i^{(\epsilon)}\in\mathbb{R}^{|\mathcal{S}_i|\times|A_i|}_{\geq 0}$ with $Y^{(\epsilon)}_i({S_i,a_i}) := \rho_i(S_i)\tilde \pi_i^{(\epsilon)}(a_i\mid S_i)$, where $\tilde \pi_i^{(\epsilon)}$ is a conditional distribution over actions defined in Equation~\eqref{eqn:conditional}. Then for every $\epsilon\in(0,1)$, the matrix $Y_i^{(\epsilon)}$ has row sums $\rho_i$, column sums $\tilde \mu_i^{(\epsilon)}$, and $\operatorname{supp}(Y_i^{(\epsilon)})=\operatorname{supp}(X_i)$.     
    For the row player, the above construction gives us matrices $X_1$ and $Y^{(\epsilon)}_1$ given by 
    \begin{equation*}
        X_1 = \begin{pmatrix}
              1 & 1 & 1 \\
              0 & 1 & 1
        \end{pmatrix}
        \text{ and }
        Y^{(\epsilon)}_1 = 
        \begin{pmatrix}
            \frac{3-\epsilon}{12} & \frac{\epsilon}{6} & \frac{3-\epsilon}{12} \\
            0 & \frac{2-\epsilon}{4} & \frac{\epsilon}{4}
        \end{pmatrix}.
    \end{equation*}
    
    Theorem~\ref{theorem:matrix_scaling} guarantees that there exist vectors $u_i^{(\epsilon)}\in\mathbb R_{>0}^{|\mathcal{S}_i|}$ and $v_i^{(\epsilon)}\in\mathbb R_{>0}^{|A_i|}$ such that
    $P_i^{(\epsilon)} = \operatorname{diag}\big(u_i^{(\epsilon)}\big) X_i \operatorname{diag}\big(v_i^{(\epsilon)}\big)$
    has row sums $\rho_i$ and column sums $\tilde \mu^{(\epsilon)}_i$. In our case, by solving the system
    \begin{equation*}
    \label{eqn:quadraticsystem}
        \begin{cases}
            \operatorname{diag}(u^{(\epsilon)}_1)X_1\operatorname{diag}(v^{(\epsilon)}_1)\mathbf{1} = Y^{(\epsilon)}_1\mathbf{1}\\
            \operatorname{diag}(v^{(\epsilon)}_1)X_1^\top\operatorname{diag}(u^{(\epsilon)}_1)\mathbf{1} = {Y^{(\epsilon)}_1}^\top\mathbf{1},
        \end{cases}
    \end{equation*}
    we get 
    \begin{equation*}
    u^{(\epsilon)}_1 = \lambda\left(1, \frac{6}{3 + \epsilon}\right) \text{ and }
    v^{(\epsilon)}_1 = \frac{1}{12\lambda}\left(3-\epsilon, \frac{(6 - \epsilon)(3 + \epsilon)}{\epsilon + 9}, \frac{(3 + 2\epsilon)(3 + \epsilon)}{\epsilon + 9}\right),
    \end{equation*}
    for any constant $\lambda > 0$. 
    Setting $\lambda = \frac{1}{2}$ and letting $\epsilon\to 0$,
    we get $v^{(\epsilon)}_1\to w_1 = \left(\frac{1}{2}, \frac{1}{3}, \frac{1}{6}\right)$. With similar calculations, we obtain $w_2 = \left(\frac{2}{5}, \frac{2}{5}, \frac{1}{5}\right)$. With $w_1$ and $w_2$, we have $\pi_1(S_{11}) = \left(\frac{1}{2}, \frac{1}{3}, \frac{1}{6}\right)$, $\pi_1(S_{12}) = \left(0, \frac{2}{3}, \frac{1}{3}\right)$, $\pi_2(S_{21}) = \left(\frac{2}{3}, 0, \frac{1}{3}\right)$, and $\pi_2(S_{22}) = \left(0, \frac{2}{3}, \frac{1}{3}\right)$, which are consistent with the marginal distribution $\mu^*$ and also yield expected payoffs of 0 for each player. Therefore, $\pi$ is also a Nash equilibrium.

\subsection{Proof of Proposition \ref{prop:SINash}}
\label{app:proof-ranking-regret-ne}
\SINashConvergence*
\begin{proof}
We first show that in GSAS, rankings contain a best fixed response in hindsight. Fix a player $i$ and define
\[
    G_i(a_i)
    :=
    \sum_{t=1}^T U_i(a_i;\pi_{-i}^t),
    \quad a_i\in A_i.
\]

For every $\pi_i$,
\[
\begin{aligned}
    \sum_{t=1}^T U_i(\pi_i,\pi_{-i}^t)
    &=
    \sum_{S_i\in\mathcal S_i}\rho_i(S_i)
    \sum_{a_i\in S_i}
    \pi_i(a_i\mid S_i)G_i(a_i)                                      \\
    &\le
    \sum_{S_i\in\mathcal S_i}\rho_i(S_i)
    \max_{a_i\in S_i}G_i(a_i).
\end{aligned}
\]

A ranking that orders the actions in decreasing order of $G_i$ attains the last expression for every $S_i$, with arbitrary tie-breaking. Therefore,
\begin{equation*}
    \max_{\tau_i\in\mathcal T_i}
    \sum_{t=1}^T U_i(\pi_i^{\tau_i},\pi_{-i}^t)
    =
    \max_{\pi_i}
    \sum_{t=1}^T U_i(\pi_i,\pi_{-i}^t).
    \label{eq:rankings-suffice}
\end{equation*}

Let $U:=U_1=-U_2$. The marginal representation of payoffs gives
\begin{align*}
    R_{T,1}^{\mathrm{rank}}+R_{T,2}^{\mathrm{rank}}
    &=
    \max_{\mu_1\in M_1}
    \sum_{t=1}^T U(\mu_1,\mu_2^t)
    -
    \min_{\mu_2\in M_2}
    \sum_{t=1}^T U(\mu_1^t,\mu_2)                                  \\
    &=
    T\left[
        \max_{\mu_1\in M_1}
        U(\mu_1,\bar\mu_2^T)
        -
        \min_{\mu_2\in M_2}
        U(\bar\mu_1^T,\mu_2)
    \right],
\end{align*}
where the second equality follows from bilinearity. The gap is the sum of the two unilateral-deviation gains at $(\bar\mu_1^T,\bar\mu_2^T)$, so each gain is at most $\epsilon_T$. Combining the above with \cref{thm:ne-implementable} then gives the result for every strategy profile implementing these marginals.
\end{proof}
\begin{cor}[Empirical-marginal residual]
\label{cor:empirical-marginal-spr}
For \(x_i\in\Delta(A_i)\), define
\[
\widehat{\operatorname{SPR}}(x_1,x_2)
:=
\max_{\mu_1\in M_1}U(\mu_1,x_2)
-
\min_{\mu_2\in M_2}U(x_1,\mu_2).
\]

Then, for every \(p\in(0,1)\), with probability at least
\(1-p\),
\[
\begin{aligned}
\widehat{\operatorname{SPR}}
(\tilde\mu_1^T,\tilde\mu_2^T)
\le
\frac{
R_{T,1}^{\mathrm{rank}}
+
R_{T,2}^{\mathrm{rank}}
}{T}+
\sqrt{\frac{2|A_1|}{pT}}
+
\sqrt{\frac{2|A_2|}{pT}}.
\end{aligned}
\]
\end{cor}

\begin{proof}
By concentration of the empirical marginals and a union
bound, with probability at least \(1-p\),
\[
\|\tilde\mu_i^T-\bar\mu_i^T\|_1
\le
\sqrt{\frac{2|A_i|}{pT}},
\quad i\in\{1,2\}.
\]

Because \(u_i\in[-1,1]\), expected payoffs are
\(1\)-Lipschitz in each marginal under the \(\ell_1\)-norm.
Hence,
\[
\widehat{\operatorname{SPR}}
(\tilde\mu_1^T,\tilde\mu_2^T)
\le
\operatorname{SPR}
(\bar\mu_1^T,\bar\mu_2^T)
+
\sum_{i=1}^2
\|\tilde\mu_i^T-\bar\mu_i^T\|_1.
\]

The result follows from Proposition 5.3.
\end{proof}

\subsection{Counterexamples for Sleeping External and Internal Regret (Remark~\ref{rem:externalfails})}\label{appsec:counter}

\subsubsection{Sleeping External Regret}
We first define a notion of Sleeping External Regret (SE-regret) which was introduced in~\citep{blum2007external,kleinberg2010regret}. 
\begin{definition}[Sleeping External Regret]
For any action $a'_i\in A_i$, the \emph{sleeping external regret} (\emph{SE-regret}) for player $i$ is defined as:
\[
R^{\mathsf{EXT}}_{T,i}(a'_i) \coloneqq 
\mathbb{E}_{S \sim \rho}\left[\mathbb{E}_{a \sim \pi(S)}\left[\sum_{t=1}^T \bbone\{a'_i\in S_i^t\} (u_i(a'_i,a^t_{-i}) - u_i(a^t_i,a^t_{-i})) \right]\right].
\]
\end{definition}

In other words, sleeping external regret captures player $i$'s expected payoff gain from switching to $a'_i$ whenever it is available, regardless of the action originally chosen. In the case where a player's SE-regret grows sublinearly, i.e., $\max_{a'_i} R^{\mathsf{EXT}}_{T,i}(a'_i) = o(T)$, they are said to have \emph{no-SE-regret}.

\paragraph{Counterexample for SE-regret} We construct a single-player GSAS $\calG=\left(\calG_{\mathrm{orig}},\mathcal{S},\rho\right)$ where $\calG_{\mathrm{orig}}$ has three actions $a_1, a_2, a_3$. Let $\mathcal{S} = \{S_{1}=\{a_1,a_2\}, S_2=\{a_2,a_3\} \}$ with $\rho(S_1) = \rho(S_2) = 0.5$. The payoff function is given as $u(a_1) = 0.01$, $u(a_2) = 0.02$ and $u(a_3) = 1$. Suppose the player plays uniformly in $S_1$ and plays $a_3$ with probability 1 in $S_2$. The SE-regret per round for each action is as follows:
\begin{align*}
    \frac{R^{\mathsf{EXT}}_T(a_1)}{T} &= \frac12(0.01-0.015) = -0.0025,\\
    \frac{R^{\mathsf{EXT}}_T(a_2)}{T} &= \frac12\left((0.02-1) + (0.02-0.015)\right) = -0.4875,\\
    \frac{R^{\mathsf{EXT}}_T(a_3)}{T} &= 0.
\end{align*}

Thus for each action $a$, $[R^{\mathsf{EXT}}_T(a)]_+=0$; hence it has no-SE-regret. But the strategy is not a Nash equilibrium, since a profitable deviation would be to play $a_2$ with probability 1 in $S_1$ instead; its gain is $\frac12(0.02-0.015)=0.0025$ per round. Therefore, no-SE-regret does not suffice to guarantee convergence to Nash equilibrium in GSAS.


\subsubsection{Sleeping Internal Regret}
Sleeping internal regret (SI-regret) is a stronger notion of sleeping regret introduced in~\cite{gaillard2023one} to study the sleeping bandit problem~\citep{kleinberg2010regret}.


\begin{definition}[Sleeping Internal Regret]\label{def:internalregret}
    For any pair of actions $\hat{a}_i\in A_i$ and $\hat{a}_i'\in A_i$, the \emph{sleeping internal regret} (\emph{SI-regret}) for player $i$ in $T$ rounds is 
    \begin{align*}
        R^{\mathsf{INT}}_{T,i}(\hat{a}_i\to \hat{a}_i') \coloneqq
        \mathbb{E}_{S \sim \rho}\left[\mathbb{E}_{a \sim \pi(S)}\left[\sum_{t=1}^T\mathbb{1}\{a_i^t = \hat{a}_i, \hat{a}_i'\in S_i^t\}\cdot
        \left(u_i(\hat{a}_i', a_{-i}^t) - u_i(a^t_i, a_{-i}^t)\right)\right]\right].
    \end{align*}
\end{definition}

Similar to SE-regret, if a player's SI-regret grows sublinearly for each action pair $(\hat{a}_i, \hat{a}'_i)$, i.e., $\max_{\hat{a}_i, \hat{a}'_i} R^{\mathsf{INT}}_{T,i}(\hat{a}_i\to \hat{a}_i') = o(T)$, they are said to have \emph{no-SI-regret}.
The intuition is that player $i$ does not regret not playing action $\hat{a}_i'$  (if $\hat{a}_i'$ was available) every time they played $\hat{a}_i$, for any $\hat{a}_i$, $\hat{a}'_i$. With respect to the counterexample for SE-regret above, one can easily see that the SI-regret is \emph{not} sublinear in $T$, since $R^{\mathsf{INT}}_T(a_1\to a_2)/T = 1/2\cdot1/2\cdot(0.02-0.01) = 0.0025>0$ at each round.

Analogous to the relationship between the non-sleeping variants of internal and external regret, no-SI-regret implies no-SE-regret (though the converse does not hold). This follows from the definitions of SI-regret and SE-regret:
\[
R^{\mathsf{EXT}}_{T,i}(a_i') = \sum_{a_i\in A_i} R^{\mathsf{INT}}_{T,i}(a_i\to a_i'). 
\]

Nevertheless, even though SI-regret is stronger than SE-regret, we show that minimizing SI-regret does \emph{not} suffice for Nash convergence.

\paragraph{Counterexample for SI-regret} We construct a 2p0s-GSAS $\mathcal{G}=\left(\mathcal{G}_{\mathrm{orig}},\mathcal{S},\rho\right)$. In $\mathcal{G}_{\mathrm{orig}}$, player 1 has actions $x,y,z$ while player 2 has actions $L,R$. The payoff matrix is given as
\begin{align*}
        \begin{array}
        {ccc}
        \toprule
         & L & R \\
        \midrule
        x     & 0  & 1 \\
        y & 1 & 0 \\ z  & 1  & 1 \\
        \bottomrule
        \end{array}
    \end{align*}

Let $\mathcal{S}_1=\{S_{11},S_{12}\}$, where $S_{11}=\{x,y\}$ and $S_{12}=\{x,y,z\}$, with $\rho_1(S_{11})=\rho_1(S_{12})=1/2$. Let $\mathcal{S}_2=\{S_{21}\}$, where $S_{21}=\{L,R\}$. Thus action $z$ is available to player 1 with probability $1/2$, while both actions are available for player 2.

Consider the strategy $\pi_i^t$ defined as follows:
\begin{itemize}
\item when $t$ is odd, set $\pi_1^t(S_{11})=x,\pi_1^t(S_{12})=y,\pi_2^t(S_{21})=L$;
\item when $t$ is even, set $\pi_1^t(S_{11})=y,\pi_1^t(S_{12})=x,\pi_2^t(S_{21})=R$.
\end{itemize}

Thus, on odd rounds player 1 plays $x$ from $S_{11}$ and $y$ from $S_{12}$, while player 2 plays $L$; on even rounds player 1 reverses $x$ and $y$, while player 2 plays $R$.
We now show that both players have no-SI-regret for even $T$.

For player 1, consider the deviation $x\to y$. On an odd round, player 2 plays $L$ and this deviation is triggered when $S_1^t=S_{11}$, which has probability $1/2$. Its gain is $u_1(y,L)-u_1(x,L)=1-0=1$. Similarly, on an even round, player 2 plays $R$ and this deviation is triggered when $S_1^t=S_{12}$, which has probability $1/2$. Its gain is $u_1(y,R)-u_1(x,R)=0-1=-1$. Hence they cancel out and $R^{\mathsf{INT}}_{T,1}(x\to y)=0$. Similarly $R^{\mathsf{INT}}_{T,1}(y\to x)=0$.

Now consider deviations into $z$. As $z$ is only available with $S_{12}$, we have $u_1(z,L)-u_1(y,L)=u_1(z,R)-u_1(x,R)=0$. Therefore we also have $R^{\mathsf{INT}}_{T,1}(x\to z)=R^{\mathsf{INT}}_{T,1}(y\to z)=0$.

For player 2, it is easy to show by symmetry that $R^{\mathsf{INT}}_{T,2}(L\to R)=R^{\mathsf{INT}}_{T,2}(R\to L)=0$. Therefore, we have $R^{\mathsf{INT}}_{T,1}=R^{\mathsf{INT}}_{T,2}=0$ as required.
By definition, our strategy profile $\pi_i$ induces \[\bar{\mu}_1=\left(\frac12, \frac12, 0\right),\quad\bar{\mu}_2=\left(\frac12, \frac12\right).\]

We can calculate the corresponding payoff as $U_1(\bar{\mu}_1,\bar{\mu}_2)=1/2$.
Now, consider the alternative strategy $\pi_1'$ for player 1 given as $\pi_1'(S_{11})=y,\pi_1'(S_{12})=z$, which induces \[\mu_1'=\left(0,\frac12,\frac12\right),\]
with payoff
\[\begin{aligned}U_1(\mu_1',\bar{\mu}_2)&=\frac{1}{2}\left(\frac{u_1(y,L)+u_1(y,R)}{2}\right)+\frac{1}{2}\left(\frac{u_1(z,L)+u_1(z,R)}{2}\right) \\&=\frac{1}{2}\left(\frac{1}{2}\right)+\frac{1}{2}(1)=\frac{3}{4}.\end{aligned}\]

Hence we have that
\[U_1(\mu_1',\bar{\mu}_2)-U_1(\bar{\mu}_1,\bar{\mu}_2)=\frac34-\frac12=\frac14>0,\]
and therefore $(\bar{\mu}_1,\bar{\mu}_2)$ is not a Nash equilibrium. Hence, no-SI-regret does not guarantee convergence to NE in GSAS.

\subsection{CAT-FTPL Details}
\label{app:cat-ftpl}
Using player $i$'s observed payoff, we define the observed loss $\ell_{i,t}(a_i):=\frac{1}{2}(1-u_i(a_i,a_{-i}^t))\in[0,1].$ Let $\widehat L_i^{t-1}$ be the cumulative estimated-loss vector. CAT-FTPL draws independent standard Gumbel perturbations $Z_{i,t}(a_i)$ and plays 
\begin{equation}
    a_i^t\in\argmin_{a_i\in S_i^t}
    \left\{\eta_i\widehat L_i^{t-1}(a_i)-Z_{i,t}(a_i)\right\}.
    \label{eq:cat-ftpl-choice}
\end{equation}
Equivalently, the induced strategy  for player $i$ is
\begin{equation}
    \pi_i^t(a_i\mid S_i)
    =\frac{\mathbb 1\{a_i\in S_i\}
      e^{-\eta_i\widehat L_i^{t-1}(a_i)}}
      {\sum_{b_i\in S_i}e^{-\eta_i\widehat L_i^{t-1}(b_i)}}.
    \label{eq:cat-ftpl-policy}
\end{equation}

Thus, each perturbation vector induces a random ranking, but the algorithm stores only one score per action of the GSAS. Going forward, $\pi_i^t$ denotes the induced strategy as in \eqref{eq:cat-ftpl-policy}, and $a_i^t$ denotes its sampled action.
It remains to estimate the losses of actions that are not observed at a given round. For each $a_i\in S_i^t$, draw auxiliary sets $\widetilde S_{i,t,1},\widetilde S_{i,t,2},\ldots\stackrel{\mathrm{i.i.d.}}{\sim}\rho_i$ and independent stopping indicators $B_{i,t,k}(a_i)\sim\operatorname{Bernoulli}(\gamma_i/(1+\gamma_i))$. Define the number of trials until the first success as
\begin{equation}
    H_{i,t}(a_i):=
    \min\bigl\{k\ge1:
      a_i\in\widetilde S_{i,t,k}\ \text{or}\ B_{i,t,k}(a_i)=1\bigr\}.
    \label{eq:cat-ftpl-count}
\end{equation}

Putting these elements together, we present CAT-FTPL in~\cref{alg:cat-ftpl}. 
\begin{algorithm}[ht]
\caption{CAT-FTPL for player $i$}
\label{alg:cat-ftpl}
\begin{algorithmic}[1]
\STATE \textbf{Input:} $\eta_i,\gamma_i>0$
\STATE $\widehat L_i^0\gets0\in\mathbb R^{|A_i|}_{\ge 0}$
\FOR{$t=1,\ldots,T$}
    \STATE Observe $S_i^t$, draw Gumbel perturbations and select $a_i^t$ according to
           \eqref{eq:cat-ftpl-choice}
    \STATE Observe $\ell_{i,t}(a_i)$ for every $a_i\in S_i^t$
    \STATE $\widehat\ell_{i,t}\gets0\in\mathbb R^{|A_i|}_{\ge 0}$
    \FOR{$a_i\in S_i^t$}
        \STATE Generate $H_{i,t}(a_i)$ by \eqref{eq:cat-ftpl-count} and set
        $\widehat\ell_{i,t}(a_i)\gets
         \ell_{i,t}(a_i)H_{i,t}(a_i)/(1+\gamma_i)$
    \ENDFOR
    \STATE $\widehat L_i^t\gets
           \widehat L_i^{t-1}+\widehat\ell_{i,t}$
\ENDFOR
\end{algorithmic}
\end{algorithm}

\subsection{Proof of Theorem \ref{thm:cat-ftpl}}
\label{app:cat-ftpl-proof}
\CATFTPLRegret*
\begin{proof}
Let $\mathcal F_{t-1}$ be the history before round $t$, i.e., the information generated by all play and algorithmic randomness revealed through the end of round $t-1$. Then under simultaneous play and Assumption 3.3, we have $S^t_i\perp a^t_{-i}\mid\mathcal{F}_{t-1}$. Now fix a player $i$ and omit its index. Write $d:=|A|$, $q(a):=\Pr(a\in S_t)$, and $\mathbb E_t[\cdot]:=\mathbb E[\cdot\mid\mathcal F_{t-1}]$. Let
\[
    \bar\ell_t(a):=\mathbb E_t[\ell_t(a)]
    =\frac{1-U_i(a;\pi_{-i}^t)}{2},
\]
and let $\mu_t$ and $\mu^\tau$ be the marginals induced by $\pi^t$ and ranking $\tau$, respectively. As action $a$ can be selected only when it is available, every implementable marginal obeys
\begin{equation}
    \mu(a)\le q(a).
    \label{eq:cat-ftpl-mu-q}
\end{equation}

We first record the properties of the loss estimator. Conditional on $\mathcal F_{t-1}$, $H_t(a)$ is geometric on $\{1,2,\ldots\}$ with parameter $p_a=(q(a)+\gamma)/(1+\gamma)$ and is independent of $\mathbb 1\{a\in S_t\}\ell_t(a)$. Therefore
\begin{align}
    \mathbb E_t[\hat\ell_t(a)]
      &=\frac{q(a)}{q(a)+\gamma}\bar\ell_t(a),\quad
\mathbb E_t[\hat\ell_t(a)^2]
      \le\frac{2q(a)}{(q(a)+\gamma)^2}.
    \label{eq:cat-ftpl-moments}
\end{align}

We also have the one-sided exponential-moment bound
\begin{equation}
    \mathbb E_t\left[
      e^{\gamma(\hat\ell_t(a)-\bar\ell_t(a))}
    \right]\le1.
    \label{eq:cat-ftpl-mgf}
\end{equation}

To see why this holds, let $s:=\gamma/(1+\gamma)$ and $M_a(s):=\mathbb E[e^{sH_t(a)}]$. Since $e^{sHx}\le1+x(e^{sH}-1)$ for $x\in[0,1]$,
\[
    \mathbb E_t[e^{\gamma\hat\ell_t(a)}]
    \le1+q(a)\bar\ell_t(a)(M_a(s)-1).
\]

For $q(a)>0$,
\[
    M_a(s)-1
    =\frac{e^s-1}{1-(1-p_a)e^s}
    \le\frac{\gamma}{q(a)},
\]
because $s\le\log(1+\gamma)$. The case $q(a)=0$ is immediate. Hence the above display is at most $1+\gamma\bar\ell_t(a)\le e^{\gamma\bar\ell_t(a)}$, proving \eqref{eq:cat-ftpl-mgf}.

Let $\hat L_t:=\sum_{s=1}^t\hat\ell_s$ and define the analysis potential
\[
    \Phi_t:=\mathbb E_{S\sim\rho_i}\left[
      \log\sum_{a\in S}e^{-\eta\hat L_{t-1}(a)}
    \right].
\]

By the Gumbel-max identity, the strategy induced by \eqref{eq:cat-ftpl-choice} is \eqref{eq:cat-ftpl-policy}. Thus $e^{-x}\le1-x+x^2/2$ for $x\ge0$ gives
\begin{equation}
    \Phi_{t+1}-\Phi_t
    \le-\eta\langle\mu_t,\hat\ell_t\rangle
       +\frac{\eta^2}{2}
        \sum_a\mu_t(a)\hat\ell_t(a)^2.
    \label{eq:cat-ftpl-potential-step}
\end{equation}

Moreover, $\Phi_1\le\log d$, and a ranking that orders actions by increasing $\hat L_T(a)$ minimizes the cumulative estimated loss in every availability set. Hence
\[
    \Phi_{T+1}\ge
    -\eta\min_\tau\sum_{t=1}^T
      \langle\mu^\tau,\hat\ell_t\rangle.
\]

Summing \eqref{eq:cat-ftpl-potential-step} yields, pathwise,
\begin{equation}
    \sum_t\langle\mu_t,\hat\ell_t\rangle
      -\min_\tau\sum_t\langle\mu^\tau,\hat\ell_t\rangle\le\frac{\log d}{\eta}
      +\frac{\eta}{2}\sum_{t,a}
       \mu_t(a)\hat\ell_t(a)^2.
    \label{eq:cat-ftpl-estimated-regret}
\end{equation}

By \eqref{eq:cat-ftpl-mu-q} and \eqref{eq:cat-ftpl-moments},
\begin{equation}
    \mathbb E\left[
      \sum_{t,a}\mu_t(a)\hat\ell_t(a)^2
    \right]\le2dT.
    \label{eq:cat-ftpl-variance}
\end{equation}

Now define the pathwise ranking regret in terms of conditional mean losses,
\[
    \mathcal R_T^\ell
    :=\sum_t\langle\mu_t,\bar\ell_t\rangle
      -\min_\tau\sum_t\langle\mu^\tau,\bar\ell_t\rangle.
\]

Adding and subtracting the estimated losses gives
\begin{align*}
    \mathcal R_T^\ell
    \le{}&
    \left[\sum_t\langle\mu_t,\hat\ell_t\rangle
      -\min_\tau\sum_t\langle\mu^\tau,\hat\ell_t\rangle\right]\\
    &+\sum_t\langle\mu_t,\bar\ell_t-\hat\ell_t\rangle
      +\max_a\sum_t
       (\hat\ell_t(a)-\bar\ell_t(a)).
\end{align*}

The expected second term is at most $d\gamma T$ by \eqref{eq:cat-ftpl-mu-q}--\eqref{eq:cat-ftpl-moments}. Iterating \eqref{eq:cat-ftpl-mgf}, followed by log-sum-exp and Jensen's inequality, gives
\begin{equation}
    \mathbb E\left[
      \max_a\sum_t(\hat\ell_t(a)-\bar\ell_t(a))
    \right]\le\frac{\log d}{\gamma}.
    \label{eq:cat-ftpl-uniform-error}
\end{equation}

The maximum over actions is sufficient because every ranking marginal $\mu^\tau$ belongs to $\Delta(A_i)$. Combining \eqref{eq:cat-ftpl-estimated-regret}, \eqref{eq:cat-ftpl-variance}, and \eqref{eq:cat-ftpl-uniform-error} gives
\[
    \mathbb E[\mathcal R_T^\ell]
    \le\frac{\log d}{\eta}+\eta dT+d\gamma T+\frac{\log d}{\gamma}.
\]

Finally, the normalization for loss implies $R_{T,i}^{\mathrm{rank}}=2\mathcal R_T^\ell$. This proves \eqref{eq:cat-ftpl-bound}; substituting the prescribed parameters completes the proof.
\end{proof}

\subsection{Proof of Proposition \ref{prop:cat-ftpl-high-prob-regret}}
\HighProbRegret*
\begin{proof}
First, note that ranking regret could be negative: an adaptive learner may outperform every fixed ranking in hindsight. We nevertheless only need to seek an upper-tail bound. Indeed, write $[x]_+:=\max\{x,0\}$, and for every $r\ge0$,
\[
    \{R_T^{\mathrm{rank}}>r\}
    =
    \{[R_T^{\mathrm{rank}}]_+>r\}.
\]

It therefore suffices to control the nonnegative random variable $[R_T^{\mathrm{rank}}]_+$, to which Markov's inequality can be applied. Fix a player and omit its index, and write $d:=|A|$. The proof of Theorem \ref{thm:cat-ftpl} gives, pathwise,
\begin{equation}
\frac{1}{2}R_T^{\mathrm{rank}}
\le
\frac{\log d}{\eta}
+\frac{\eta}{2}V_T+L_T+C_T,
\label{eq:cat-ftpl-proof-certificate}
\end{equation}
where
\begin{align*}
V_T&:=\sum_{t,a}\mu_t(a)\hat\ell_t(a)^2,\\
L_T&:=\sum_t\langle\mu_t,\bar\ell_t-\hat\ell_t\rangle,\\
C_T&:=\max_a\sum_t
      \bigl(\hat\ell_t(a)-\bar\ell_t(a)\bigr),\quad\bar\ell_t(a):=\mathbb E_t[\ell_t(a)].
\end{align*}

Since the first two terms on the right are nonnegative, taking positive parts gives
\begin{equation}
\frac{1}{2}[R_T^{\mathrm{rank}}]_+
\le
\frac{\log d}{\eta}
+\frac{\eta}{2}V_T+[L_T]_+ +[C_T]_+.
\label{eq:cat-ftpl-positive-certificate}
\end{equation}

We first bound $\mathbb E[V_T]$. From \eqref{eq:cat-ftpl-moments}, we have
\[
    \mathbb E_t[\hat\ell_t(a)^2]
    \le
    \frac{2q(a)}{(q(a)+\gamma)^2},
\]

Consequently,
\begin{align*}
    \mathbb E_t\left[
       \sum_{a\in A}\mu_t(a)\hat\ell_t(a)^2
    \right]
    &\le
    2\sum_{a\in A}
       \frac{\mu_t(a)q(a)}{(q(a)+\gamma)^2} \\
    &\le
    2\sum_{a\in A}
       \frac{q(a)^2}{(q(a)+\gamma)^2}
    \le 2d.
\end{align*}

Summing over $t$ and applying the tower rule therefore gives
\begin{equation}
    \mathbb E[V_T]\le 2dT.
    \label{eq:cat-ftpl-vt}
\end{equation}

To bound $\mathbb E[[L_T]_+]$, let $Y_t:=\langle\mu_t,\bar\ell_t-\hat\ell_t\rangle.$ Using \eqref{eq:cat-ftpl-mu-q} and \eqref{eq:cat-ftpl-moments},
\[
    0\le \mathbb E_t[Y_t]
    =\sum_a\mu_t(a)\frac{\gamma}{q(a)+\gamma}\bar\ell_t(a)
    \le d\gamma.
\]

Consider the martingale-difference sequence \(D_t:=Y_t-\mathbb E_t[Y_t]\). Jensen's inequality gives
\[
\begin{aligned}
    \mathbb E_t[D_t^2]
    &\le
    \sum_a\mu_t(a)\mathbb E_t[\hat\ell_t(a)^2]  \\
    &\le
    2\sum_a\frac{\mu_t(a)q(a)}{(q(a)+\gamma)^2}
    \le 2d.
\end{aligned}
\]

Therefore,
\begin{align}
    \mathbb E[[L_T]_+]
    &\le d\gamma T
      +\mathbb E\left[\left(\sum_tD_t\right)_+\right] \notag\\
    &\le d\gamma T
      +\sqrt{\sum_t\mathbb E[D_t^2]}
    \le d\gamma T+\sqrt{2dT}.
    \label{eq:cat-ftpl-positive-lt}
\end{align}

For $\mathbb E[[C_T]_+]$, define
\[
    Z_{T,a}
    :=
    \sum_{t}
    \bigl(\hat\ell_t(a)-\bar\ell_t(a)\bigr).
\]

Recall from \eqref{eq:cat-ftpl-mgf} that we have
\[
    \mathbb E_t\left[
      e^{\gamma(\hat\ell_t(a)-\bar\ell_t(a))}
    \right]\le1,
\]
and the tower rule implies
\[
    \mathbb E[e^{\gamma Z_{T,a}}]\le1.
\]

Therefore,
\[
\begin{aligned}
    \mathbb E\left[e^{\gamma[C_T]_+}\right]
    &\le
    1+\mathbb E[e^{\gamma C_T}]\\
    &\le
    1+\sum_{a\in A}\mathbb E[e^{\gamma Z_{T,a}}]
    \le d+1.
\end{aligned}
\]

Applying Jensen's inequality gives
\begin{equation}
    \mathbb E[[C_T]_+]
    \le\frac{\log (d+1)}{\gamma}
    \le\frac{\log d+1}{\gamma}.
    \label{eq:cat-ftpl-positive-ct}
\end{equation}

Combining these estimates \eqref{eq:cat-ftpl-vt}, \eqref{eq:cat-ftpl-positive-lt}, \eqref{eq:cat-ftpl-positive-ct} with \eqref{eq:cat-ftpl-positive-certificate} and setting $\eta=\gamma=\sqrt{\log d/(dT)}$ gives
\[
    \mathbb E[[R_T^{\mathrm{rank}}]_+]
    \le\left(8+\frac{2}{\log d}+2\sqrt{\frac{2}{\log d}}\right)\sqrt{dT\log d}.
\]

For $d\ge2$, we have
\[8+\frac{2}{\log d}+2\sqrt{\frac{2}{\log d}}\le8+\frac{2}{\log 2}+2\sqrt{\frac{2}{\log 2}}<16,\]
so \begin{equation*}
    \mathbb E[[R_T^{\mathrm{rank}}]_+]
    \le16\sqrt{dT\log d}.
\end{equation*}

Finally, Markov's inequality and a union bound over the players prove the proposition.
\end{proof}

\subsection{Proof of Proposition~\ref{prop:compute_w}}\label{appsec:5.7}
\begin{algorithm}[H]
\caption{Computing a compact vector for player $i$}
\label{alg:compute_w}
\begin{algorithmic}[1]
\STATE \textbf{Input:} $\eta_t$
\STATE $\theta_i^0\gets\mathbf0_{|A_i|}$, $\tilde\mu_i^0\gets\mathbf0_{|A_i|}$
\FOR{$t=1,\ldots,T$}
    \STATE Observe $S_i^t$ and play $\pi_i^t(S_i^t)$
    \STATE $\tilde\mu_i^t\gets\frac{t-1}{t}\tilde\mu_i^{t-1}+\frac1t\pi_i^t(S_i^t)$
    \STATE $G_i^t(a_i)\gets\tilde\mu_i^t(a_i)-\frac{\bbone\{a_i\in S_i^t\}\exp(\theta_i^{t-1}(a_i))}{\sum_{a_i'\in S_i^t}\exp(\theta_i^{t-1}(a_i'))}$, $\forall a_i\in A_i$
    \STATE $\theta_i^t\gets\theta_i^{t-1}+\eta_tG_i^t$
\ENDFOR
\RETURN $w_i^T=\exp(\theta_i^T)/\sum_{a_i\in A_i}\exp(\theta_i^T(a_i))$
\end{algorithmic}
\end{algorithm}
\stochasticapproximation*
\begin{proof}
Suppress the player indices and consider
\[
 p_\theta(a\mid S):=\frac{\bbone\{a\in S\}e^{\theta(a)}}{\sum_{b\in S}e^{\theta(b)}},
 \qquad \Phi(\theta):=\mathbb E_{S\sim\rho}\log\sum_{a\in S}e^{\theta(a)}.
\]

Then $\nabla\Phi(\theta)=\hat\mu(\exp\theta)$. For each $S$, the function $\log\sum_{a\in S}e^{\theta(a)}$ has Hessian $\operatorname{diag}(p)-pp^\top$ for $p=p_\theta(\cdot\mid S)$, with eigenvalues in $[0,1]$, so $\Phi$ is convex and its gradient is $1$-Lipschitz.
The bounded martingale differences $\pi^t(\cdot\mid S^t)-\mu^t$ satisfy
\[
 \Pr\left(\|\tilde\mu^t-\bar\mu^t\|_\infty>\varepsilon\right)
 \le2|A|e^{-t\varepsilon^2/2}
\]
by the Azuma--Hoeffding inequality~\citep{azuma1967weighted,hoeffding1963probability} and a union bound over actions. Since these probabilities are summable for every $\varepsilon>0$, $\tilde\mu^t-\bar\mu^t\to0$ almost surely. Each $\bar\mu^t$ belongs to the closed set $M$, and $\mu^*\in M$ almost surely. Consider the case $\mu^*\in M$ and define
\[
 f(\theta):=\Phi(\theta)-\langle\mu^*,\theta\rangle,
 \qquad f_*:=\inf_\theta f(\theta).
\]

Since $\mu^*$ is implementable, $\langle\mu^*,\theta\rangle\le\mathbb E_{S\sim\rho}\max_{a\in S}\theta(a)\le\Phi(\theta)$, so $f\ge0$, and $\|\nabla f\|_2\le\sqrt2$ since $\nabla\Phi$ and $\mu^*$ are probability distributions. We claim that, for every $\delta>0$, some $c_\delta>0$ satisfies
\begin{equation}\label{eq:rm-small-gradient}
 f(\theta)\ge f_*+\delta\quad\implies\quad\|\nabla f(\theta)\|_2\ge c_\delta.
\end{equation}

To see why this is the case, towards a contradiction, take $\theta_k$ with $f(\theta_k)\ge f_*+\delta$ and $\nabla f(\theta_k)\to0$.
The finitely many coordinates $p_{\theta_k}(a\mid S)\in[0,1]$ admit a subsequence with $p_{\theta_k}(\cdot\mid S)\to q_S$ for every $S$. Since $\sum_S\rho(S)q_S=\mu^*$, for every $\theta$,
\begin{equation}\label{eq:rm-entropy-identity}
 \lim_k[f(\theta_k)-f(\theta)]
 =\sum_{S,a}\rho(S)q_S(a)\log\frac{p_\theta(a\mid S)}{q_S(a)}\le0.
\end{equation}

Here the summation is over $\rho(S)>0$ and $a\in S$ with $q_S(a)>0$. Then, the inequality in Equation~\eqref{eq:rm-entropy-identity} follows from $\log x\le x-1$ and $\sum_{a:q_S(a)>0}p_\theta(a\mid S)\le1$.
Thus $f(\theta)\ge f_*+\delta$ for every $\theta$, contradicting the definition of $f_*$.

Let $\mathcal H_t$ be the history and chosen strategy before $S^t$ is drawn, and set
\[
 \xi_t:=\nabla\Phi(\theta^{t-1})-p_{\theta^{t-1}}(\cdot\mid S^t).
\]

Since $\mathbb E[\xi_t\mid\mathcal H_t]=0$, $\|\xi_t\|_2\le\sqrt2$, and $\sum_t\eta_t^2<\infty$, the partial sums of $\sum_t\eta_t\xi_t$ are martingales with bounded second moments and thus converge almost surely.
Fix an outcome for which these partial sums converge and $\tilde\mu^t\to\mu^*\in M$. For the remainder of the proof, define
\[\vartheta^t:=\theta^t+\sum_{k>t}\eta_k\xi_k.
\]

The series tail tends to zero, so $\vartheta^t-\theta^t\to0$. Subtracting $\eta_t\xi_t$ from the update in~\cref{alg:compute_w} gives
\[
 \vartheta^t-\vartheta^{t-1}
 =\theta^t-\theta^{t-1}-\eta_t\xi_t
 =\eta_t[\tilde\mu^t-\nabla\Phi(\theta^{t-1})],
\]
so $\|\vartheta^t-\vartheta^{t-1}\|_2\le\sqrt2\eta_t$.

Set $\varepsilon_t:=\|\tilde\mu^t-\mu^*\|_2+\|\vartheta^{t-1}-\theta^{t-1}\|_2\to0$.
The gradient bound and its Lipschitz property give
\begin{align}
 f(\vartheta^t)-f(\vartheta^{t-1})
 &\le-\eta_t\|\nabla f(\vartheta^{t-1})\|_2^2+\sqrt2\eta_t\varepsilon_t+\eta_t^2,\label{eq:rm-descent}\\
 |f(\vartheta^t)-f(\vartheta^{t-1})|&\le2\eta_t\to 0.\label{eq:rm-step-change}
\end{align}

Fix some $\delta>0$. For all sufficiently large $t$, $\sqrt2\varepsilon_t+\eta_t\le c_\delta^2/2$ and $2\eta_t\le\delta$.
By~\eqref{eq:rm-descent} and~\eqref{eq:rm-small-gradient}, whenever $f(\vartheta^{t-1})\ge f_*+\delta$, the subsequent value $f(\vartheta^{t})$ decreases by at least $c_\delta^2\eta_t/2$.
Hence for sufficiently large $t$, the sequence $\{f(\vartheta^t)\}$ must fall below $f_*+\delta$: otherwise, summing these decreases contradicts $f\ge f_*$ because $\sum_t\eta_t=\infty$.
It then stays below $f_*+2\delta$: indeed, a step starting below $f_*+\delta$ increases the value by at most $\delta$ by~\eqref{eq:rm-step-change} and a step starting at or above $f_*+\delta$ decreases it.
Since $\delta>0$ is arbitrary and $0 \le f(\vartheta^t) - f_* < 2\delta$, we conclude $f(\vartheta^t)\to f_*$.

Finally, the Lipschitz gradient gives
\[
 f_*\le f(\theta-\nabla f(\theta))
 \le f(\theta)-\tfrac12\|\nabla f(\theta)\|_2^2.
\]

Hence $\nabla f(\vartheta^t)\to0$. Since $\vartheta^t-\theta^t\to0$,
$\hat\mu(w^T)=\nabla\Phi(\theta^T)\to\mu^*$ almost surely, as required.
\end{proof}

\subsection{Proof of Theorem~\ref{thm:mainthm}}\label{appsec:robustSA}

\begin{algorithm}[ht]
\caption{Finite-time RSA for player $i$}
\label{alg:compute_w_rsa}
\begin{algorithmic}[1]
\STATE \textbf{Input:} $\tilde\mu_i\in\Delta(A_i)$
\STATE $\theta_i^1\gets\mathbf0_{|A_i|}$, $\lambda\gets T^{-1/2}$, $s\gets\lfloor T/2\rfloor+1$ 
\FOR{$t=1,\ldots,T$}
    \STATE Draw $S_i^t\sim\rho_i$
    \STATE $\eta_t\gets2/(\lambda(t+1))$ 
    \STATE $G_i^t(a_i)\gets\tilde\mu_i(a_i)-\frac{\bbone\{a_i\in S_i^t\}\exp(\theta_i^t(a_i))}{\sum_{a_i'\in S_i^t}\exp(\theta_i^t(a_i'))}-\lambda\theta_i^t(a_i)$, $\forall a_i\in A_i$
    \STATE $\theta_i^{t+1}\gets\theta_i^t+\eta_tG_i^t$
\ENDFOR
\STATE $\tilde\theta_i\gets\sum_{t=s}^T(\eta_t\theta_i^t)/\sum_{t=s}^T\eta_t$ 
\RETURN $\tilde w_i=\exp(\tilde\theta_i)/\sum_{a_i\in A_i}\exp(\tilde\theta_i(a_i))$
\end{algorithmic}
\end{algorithm}

\cite{nemirovski1978cezari} initially proposed the use of Cesàro means to avoid non-convergence or slow convergence for Lipschitz convex functions in stochastic approximation, a method they referred to as \emph{robust stochastic approximation} (RSA). 
Algorithm~\ref{alg:compute_w_rsa} adapts this framework to a quadratically regularized objective with decreasing stepsize; to accommodate this, we employ a trick known as \emph{suffix averaging}, which was initially proposed in~\cite{rakhlin2012making}. We mostly follow the analysis of~\citet{nemirovski2009robust}, using strong convexity to bound the error in minimizing the regularized objective after $T$ timesteps. We then account for the quadratic regularizer and the difference between the empirical and true marginals to obtain the following bound.


\begin{lemma}\label{lem:rsaconvergence_new}
Fix a player $i$ and an input $\tilde\mu_i\in\Delta(A_i)$ to \cref{alg:compute_w_rsa}. For every implementable marginal $\mu_i\in M_i$, its output satisfies
\begin{equation}\label{eqn:cesaro_projection_bound_new}
 \mathbb E\bigl[\|\hat\mu_i(\tilde w_i)-\mu_i\|_2^2\bigr]
 \le \frac{192+4\log|A_i|}{\sqrt T}
 +2\|\tilde\mu_i-\mu_i\|_2^2.
\end{equation}
\end{lemma}
\begin{proof}
Suppress the player indices and use $\Phi$ from the preceding proof. The function
\[
 F_\lambda(\theta):=\Phi(\theta)-\langle\tilde\mu,\theta\rangle
 +\frac\lambda2\|\theta\|_2^2
\]
is coercive and $\lambda$-strongly convex, so its unique minimizer $\theta^*$ satisfies
\begin{equation}\label{eq:rsa-optimality}
 \nabla\Phi(\theta^*)-\tilde\mu+\lambda\theta^*=0.
\end{equation}

Since $0\le\lambda\eta_t\le1$, the update
\[
 \theta^{t+1}=(1-\lambda\eta_t)\theta^t
 +\lambda\eta_t\frac{\tilde\mu-p_{\theta^t}(\cdot\mid S^t)}{\lambda}
\]
and $\theta^1=0$ imply that $\lambda\|\theta^t\|_2\le\sqrt2$ by induction. Thus $-G^t$ is an unbiased stochastic gradient of $F_\lambda$ with norm at most $2\sqrt2$.

Let $D_t=\mathbb E\|\theta^t-\theta^*\|_2^2$. Expanding the squared distance norm and using strong convexity gives
\begin{equation}\label{eq:rsa-objective-1}
 2\eta_t\mathbb E[F_\lambda(\theta^t)-F_\lambda(\theta^*)]
 \le(1-\lambda\eta_t)D_t-D_{t+1}+8\eta_t^2.
\end{equation}

Since the LHS of the above is nonnegative, using $\eta_t=2/[\lambda(t+1)]$ yields
$D_t\le32/[\lambda^2(t+1)]$ by induction. The initial bound $D_1\le2/\lambda^2$ follows from~\eqref{eq:rsa-optimality}.

For the averaging bound, note that $-\lambda\eta_tD_t\le0$,
so omitting this term only increases the right-hand side. Using the suffix averaging trick from~\cite{rakhlin2012making}, we sum~\eqref{eq:rsa-objective-1} over $t=s,\ldots,T$ where $s=\lfloor T/2\rfloor+1$, which gives, by convexity
\begin{equation}\label{eq:rsa-objective}
 \mathbb E[F_\lambda(\tilde\theta)-F_\lambda(\theta^*)]
 \le\frac{D_s+8\sum_{t=s}^T\eta_t^2}{2\sum_{t=s}^T\eta_t}
 \le\frac{48}{\lambda T}.
\end{equation}

The last step uses the inequalities $s+1\ge(T+2)/2$,
$\sum_{t=s}^T\eta_t\ge T/[\lambda(T+1)]$, and
$\sum_{t=s}^T\eta_t^2\le 2\sum_{t=s}^T\eta_t/[\lambda(s+1)]$.

For any feasible $\mu$, we have $\Phi(\theta)-\langle\mu,\theta\rangle\ge0$ and $\Phi(0)\le\log |A_i|$.
Using~\eqref{eq:rsa-optimality} and $F_\lambda(\theta^*)\le F_\lambda(0)$ therefore gives
\begin{align*}
 \|\nabla\Phi(\theta^*)-\mu\|_2^2
 &=\|\tilde\mu-\mu\|_2^2+\lambda^2\|\theta^*\|_2^2
 -2\lambda\langle\tilde\mu-\mu,\theta^*\rangle\\
 &\le\|\tilde\mu-\mu\|_2^2
 +2\lambda[\Phi(0)-\Phi(\theta^*)+\langle\mu,\theta^*\rangle]\\
 &\le\|\tilde\mu-\mu\|_2^2+2\lambda\log |A_i|.
\end{align*}

Convexity and $1$-smoothness of $\Phi$, together with~\eqref{eq:rsa-optimality}, imply
\begin{align*}
 \|\nabla\Phi(\theta)-\nabla\Phi(\theta^*)\|_2^2
 &\le2[\Phi(\theta)-\Phi(\theta^*)
 -\langle\nabla\Phi(\theta^*),\theta-\theta^*\rangle]\\
 &\le2[F_\lambda(\theta)-F_\lambda(\theta^*)].
\end{align*}

Combining these bounds with~\eqref{eq:rsa-objective} and the triangle inequality yields
\[
 \mathbb E\|\nabla\Phi(\tilde\theta)-\mu\|_2^2
 \le\frac{192}{\lambda T}+4\lambda\log |A_i|
 +2\|\tilde\mu-\mu\|_2^2.
\]

Substituting $\lambda=T^{-1/2}$ and $\hat\mu(\tilde w)=\nabla\Phi(\tilde\theta)$ gives the statement as required.
\end{proof}

Now we are ready to prove Theorem~\ref{thm:mainthm}.
\mainthmnormalform*

\begin{proof}
Let $d=\max_i|A_i|$.
The differences $\pi_i^t(\cdot\mid S_i^t)-\mu_i^t$ have conditional mean zero, so cross terms from different timesteps vanish in expectation. Since each $\pi_i^t(\cdot\mid S_i^t)$ is a probability vector, it holds that
\[
\mathbb E\|\pi_i^t(\cdot\mid S_i^t)-\mu_i^t\|_2^2
=\mathbb E\|\pi_i^t(\cdot\mid S_i^t)\|_2^2
-\mathbb E\|\mu_i^t\|_2^2
\le 1.
\]

Thus, expanding the squared norms gives
\[
\mathbb E\sum_{i=1}^2\|\tilde\mu_i^T-\bar\mu_i^T\|_2^2
=\frac1{T^2}\sum_{i=1}^2\sum_{t=1}^T
\mathbb E\|\pi_i^t(\cdot\mid S_i^t)-\mu_i^t\|_2^2
\le\frac2T.
\]

We can apply~\cref{lem:rsaconvergence_new} conditionally on the CAT-FTPL history with inputs $\tilde\mu_i^T$ and $\bar\mu_i^T$. Taking expectations yields
\[
 \mathbb E\sum_{i=1}^2\|\hat\mu_i(\tilde w_i)-\bar\mu_i^T\|_2^2
 \le O\left(\frac{1+\log d}{\sqrt T}\right).
\]

Then, Markov's inequality, with failure probability $p/2$, gives
\begin{equation}\label{eq:rsa-high-prob-marginals}
 \sum_{i=1}^2\|\hat\mu_i(\tilde w_i)-\bar\mu_i^T\|_2^2
 \le O\left(\frac{1+\log d}{p\sqrt T}\right).
\end{equation}

Meanwhile, setting $p_1=p_2=p/4$ in~\cref{prop:cat-ftpl-high-prob-regret}, and combining with \cref{prop:SINash} give
\[
 \operatorname{SPR}(\bar\mu_1^T,\bar\mu_2^T)
 \le O\left(\frac{\sqrt{d\log d}}{p\sqrt T}\right)
\]
with failure probability at most $p/2$. Note that both bounds hold with probability at least $1-p$ by a union bound.

Since payoffs lie in $[-1,1]$, for every $\nu_1\in M_1$ and
$\nu_2\in M_2$,
\begin{align*}
U\bigl(\nu_1,\hat\mu_2(\tilde w_2)\bigr)
&\le U(\nu_1,\bar\mu_2^T)
+\|\hat\mu_2(\tilde w_2)-\bar\mu_2^T\|_1,\\
U\bigl(\hat\mu_1(\tilde w_1),\nu_2\bigr)
&\ge U(\bar\mu_1^T,\nu_2)
-\|\hat\mu_1(\tilde w_1)-\bar\mu_1^T\|_1.
\end{align*}

Taking the maximum over $\nu_1$ in the first inequality and
the minimum over $\nu_2$ in the second then subtracting gives, by definition of SPR,
\begin{equation}\label{eq:sprbound}
\operatorname{SPR}\bigl(\hat\mu_1(\tilde w_1),\hat\mu_2(\tilde w_2)\bigr)
\le
\operatorname{SPR}(\bar\mu_1^T,\bar\mu_2^T)
+\|\hat\mu_1(\tilde w_1)-\bar\mu_1^T\|_1
+\|\hat\mu_2(\tilde w_2)-\bar\mu_2^T\|_1.
\end{equation}

Recall that for every $a,b,c$ we have \((a+b+c)^2\le2a^2+2(b+c)^2\le2a^2+4b^2+4c^2\). Moreover, applying Cauchy--Schwarz we get that $\|\hat\mu_i(\tilde w_i)-\bar\mu_i^T\|_1 \leq \sqrt{d}\|\hat\mu_i(\tilde w_i)-\bar\mu_i^T\|_2$. Hence, it follows that
\begin{align*}\operatorname{SPR}\bigl(\hat\mu_1(\tilde w_1),\hat\mu_2(\tilde w_2)\bigr)^2&\le2\operatorname{SPR}(\bar\mu_1^T,\bar\mu_2^T)^2+4d\sum_{i=1}^2\|\hat\mu_i(\tilde w_i)-\bar\mu_i^T\|_2^2\\
&\le O\left(\frac{d(1+\log d)}{p\sqrt T}
+\frac{d\log d}{p^2T}\right).
\end{align*}

This proves that the stated bound holds with probability at least $1-p$.
\end{proof}



\section{Additional Experiments}\label{appsecs:experiments}
\subsection{Game definitions}\label{appsec:experimentalgamedefinitions}
In this section we formally define the games used in our experiments.

\begin{definition}[Random 2p0s-GSAS]\label{def:randomgsas}
A random GSAS has \(\lvert A_i \rvert = n\) actions for each player, with payoffs given by \[A_{i,j} \sim \mathrm{Uniform}[-1, 1],\]
and for each action we sample \(p_{a_i} \sim \mathrm{Uniform}[3/10, 5/10]\) when generating the game. Sampling \(S \sim\rho\) then, independently for each player, selects actions \(a_i\) to be available independently with probability \(p_{a_i}\) and redraws if no action is present.
\end{definition}

\begin{definition}[\(n \times n\) Random Biased Support (``RBS'') GSAS]\label{def:RBSgsas}
An \(n \times n\) RBS GSAS has \(\lvert A_i \rvert = n\) actions for each player, with payoffs given by \[A_{i,j} \sim \mathrm{Uniform}[-1, 1],\]
except that the upper-left block has a matching pennies structure given by \(A_{1,1}=A_{2,2}=1\) and \(A_{1,2}=A_{2,1}=0\) and that the first two actions which are always available. For the other actions we sample \(p_{a_i} \sim \mathrm{Uniform}[1/100, 2/100]\). Sampling \(S \sim \rho\) then, independently for each player, selects actions \(a_i\) to be available independently with probability \(p_{a_i}\) and redraws if no action is present.
\end{definition}

\begin{definition}[Checkerboard \(n \times n\) Matching Pennies (``Checkerboard MP'')]\label{def:checkermp}
    A checkerboard $n\times n$ MP has \(\lvert A_i \rvert = n\) actions for each player, with payoffs given by \[A_{i,j} \coloneqq \begin{cases}
1 & \text{if } i \equiv j \pmod 2,\\
-1 & \text{otherwise}.\end{cases}\]
Action availabilities for player 1 are drawn uniformly at random from\[\{\{i\} \cup \{2,4,\dots,n\} \mid i \in\{1,3,\dots,n-1\}\},\]
and player 2 always has access to all actions. Here $n$ is required to be even.
\end{definition}
Intuitively, player 1 wins if they pick an action with the same parity as player 2, and player 1 has access to all even actions and a single odd action chosen uniformly at random, while player 2 has access to all actions.

\begin{definition}[Biased \(n \times n\) Rock-Paper-Scissors (``Biased RPS'')]\label{def:biased-rps}
A biased \(n \times n\) RPS has \(\lvert A_i \rvert=n\ge3\) actions for each player, with payoffs given by
\[A_{i,j} \coloneqq \begin{cases}
-1 &\text{if }j \equiv i+1 \pmod n,\\
1 &\text{if }j \equiv i-1 \pmod n,\\
0 &\text{otherwise}.
\end{cases}\]
With probability \(2/n\), player 1 only has their first action available; otherwise player 1 has all their actions available. Player 2 always has all actions available.
\end{definition}

This is a special case of the generalized \(n \times n\) Rock-Paper-Scissors defined by~\citet{lazarsfeld2025fast}.

\begin{definition}[Biased \(n \times n\) matching pennies (``Biased MP'')]\label{def:biased-mp}
A biased \(n \times n\) MP has \(\lvert A_i \rvert = n\ge3\) actions for each player, with payoffs given by
\[A_{i,j} := \begin{cases}
1 &\text{if }i = j,\\
0 &\text{otherwise}.
\end{cases}\]
Player 1 has with probability \(4/5\) only actions \{\(a_1, \ldots, a_{\lfloor 3n/5 \rfloor}\}\) available, and with probability \(1/5\) only actions \(\{a_{\lfloor 2n/5 \rfloor}, \ldots,  a_n\}\) available. Player 2 always has all actions available.
\end{definition}

\subsection{Additional Experimental Details}\label{appsec:experimentaldetails}
All experiments were run on a 2021 MacBook Pro with 32 GB of RAM and an 8-core Apple M1 Pro chip. The Gurobi commercial solver was allowed to use any number of threads. Gurobi optimizer version 13.0.1 build v13.0.1rc0 (mac64[arm] - Darwin 24.6.0 24G90) was used. Python version 3.14.2 with NumPy version 2.5.3 and SciPy version 1.18.1 were used. Central 95\% intervals, where reported, were computed using \texttt{numpy.quantile}. To ensure reproducibility, all randomness used in the algorithms was generated by a NumPy random generator seeded with the current run number, i.e., if an experiment performed 100 runs, the \(r\)th run seeded the random number generator with \(r\).

\subsubsection{Experiment 1: Comparison with LP solver}\label{appsubsec:exp1_additional_details}
To provide a suitable comparison between CAT-FTPL and Gurobi, we first construct the sequence-form representation~\citep{von1996efficient} of the GSAS, which encodes all possible action subsets in the game, and apply Gurobi's LP solver to this expanded form, recording the wallclock convergence time.
For CAT-FTPL, we use the doubling trick to find a time horizon \(T \geq 1000\) with \(\mathrm{SPR}(\tilde{\mu_1}, \tilde{\mu_2}) \leq 0.05\) and report the wallclock time of~\cref{alg:cat-ftpl} for the final run.
We repeat the experiment 20 times and show the average and range across all runs in \cref{fig:runtime_experiment}.
As described in~\cref{sec:experiments}, with a 360-second time budget, Gurobi could solve a GSAS of size \(11 \times 11\), via an LP with 13313 variables and a linear constraint matrix of size \(13313 \times 13313\) with 126904320 nonzero entries, taking $\approx293$ seconds to solve on average.
CAT-FTPL was able to obtain low SPR for much larger games (\(350 \times 350\)) within the same time budget.

Nevertheless, we reiterate that comparing the \emph{solution quality} of these methods is moot: our method is an online regret-minimizing algorithm, and so the stopping criterion can be thought of as a proxy for terminating when the time-averaged ranking regret is low enough. This is not comparable at all to the interior point methods used by Gurobi's LP solver. The experiment is specifically designed to investigate the runtime difficulties that arise due to exponentially many action availability sets in GSAS, which is corroborated in~\cref{fig:runtime_experiment}. For a more direct comparison with some implications on solution quality, it is instructive for the reader to consider the experiment on multi-scale MWU in~\cref{appsubsec:multiscalemwu}. Since multi-scale MWU is an \emph{online} algorithm that is designed for Bayesian games, we can to some extent compare the quality of the solutions obtained.

Some additional remarks are in order. In \cref{fig:runtime_experiment}, we plot the \emph{wallclock} runtime reported by Gurobi, which does \emph{not} include the time to construct the sequence-form representation or the time to build the Gurobi model, which in practice can also be slow for large games. In other words, CAT-FTPL significantly outperforms Gurobi for large games, even without considering the additional preprocessing required for Gurobi. In line with our theoretical results, we run CAT-FTPL with \(\eta_i=\gamma_i = \sqrt{\frac{\log\lvert A_i\rvert}{\lvert A_i \rvert T}}\).

The sequence-form linear program has \(2^n + n2^{n-1}+1\) variables, and a constraint matrix of the same dimension, which quickly becomes infeasible to solve as observed.

\subsubsection{Experiment 2: Convergence in large GSAS}\label{appsubsec:exp2_additional_details}
We find that using smaller \(\lambda\) values in \cref{alg:compute_w_rsa} induces faster convergence in practice, and so run \cref{alg:compute_w_rsa} with \(\lambda = \frac{0.1}{\sqrt{T}}\). We run CAT-FTPL with the parameters from~\cref{thm:cat-ftpl}, i.e., using \(\eta_i =\gamma_i = \sqrt{\frac{\log \lvert A_i \rvert}{\lvert A_i \rvert T}}\). All runs of~\cref{alg:compute_w_rsa} for a given game use the same time-average marginal \((\tilde\mu_1, \tilde\mu_2)\), computed as the output of a CAT-FTPL run with time horizon \(T=120000\), to isolate the error contributed to the SPR to only that from the compact representation procedure~\cref{alg:compute_w_rsa}.

In~\cref{fig:regret_plots_app}, we show further regret plots on several additional 2p0s-GSAS: \cref{ex:rps}, biased \(50 \times 50\) Rock-Paper-Scissors (\cref{def:biased-rps}) and biased \(50 \times 50\) matching pennies (\cref{def:biased-mp}). We show the saddle point residual on these games for the CAT-FTPL marginals in~\cref{fig:spr_plots_marginals_app} and of the robust averaging \(w_i^t\) in~\cref{fig:spr_plots_w_app}. All results match our theoretical bounds as in~\cref{sec:experiments}.

\begin{figure}[ht]
    \centering
    \begin{subfigure}{\linewidth}
        \centering
        \includegraphics[width=\linewidth]{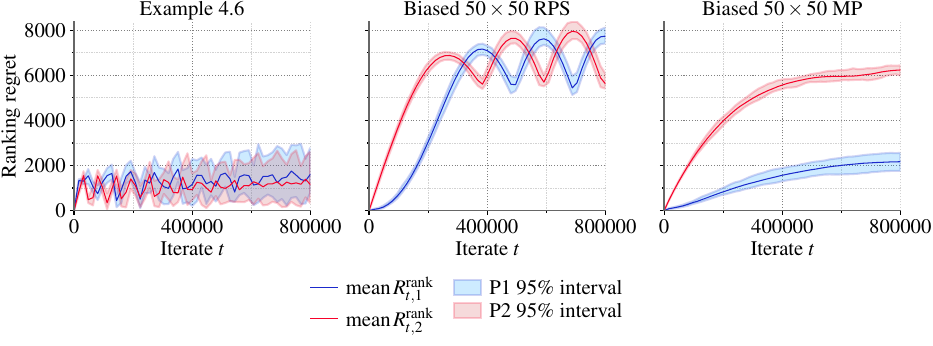}
        \caption{\(R^\mathrm{rank}_{t,i}\) for \(T=800000\). Mean and 95\% central interval from 20 runs shown. Note the theoretical bounds from~\cref{thm:cat-ftpl} and~\cref{prop:cat-ftpl-high-prob-regret} apply only to the final iterate at \(T=800000\).}
        \label{fig:regret_plots_app_iterates}
    \end{subfigure}
    \begin{subfigure}{\linewidth}
        \centering
        \includegraphics[width=\linewidth]{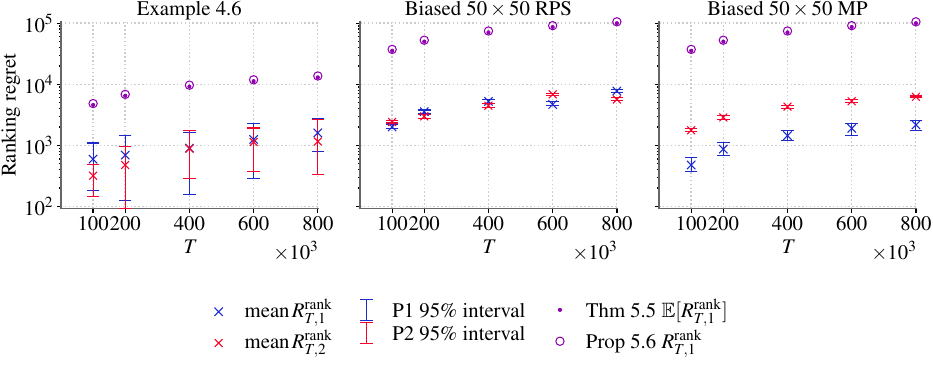}
        \caption{Ranking regret of final iterate of CAT-FTPL for various \(T\). Theoretical bounds from from~\cref{thm:cat-ftpl} and~\cref{prop:cat-ftpl-high-prob-regret} are also shown. Results from 20 runs for each \(T\) shown.}
        \label{fig:regret_plots_app_final_iterate}
    \end{subfigure}
    \caption{Ranking regret of CAT-FTPL for three 2p0s-GSAS. Average, central 95\% interval and theoretical expected/high-probability bounds on maximum ranking regret are shown. We observed \(R^\mathrm{rank}_{t,i}\) to be much lower in practice than the bound in \cref{thm:cat-ftpl}.}
    \label{fig:regret_plots_app}
\end{figure}

\begin{figure}[ht]
    \centering
    \includegraphics[width=\linewidth]{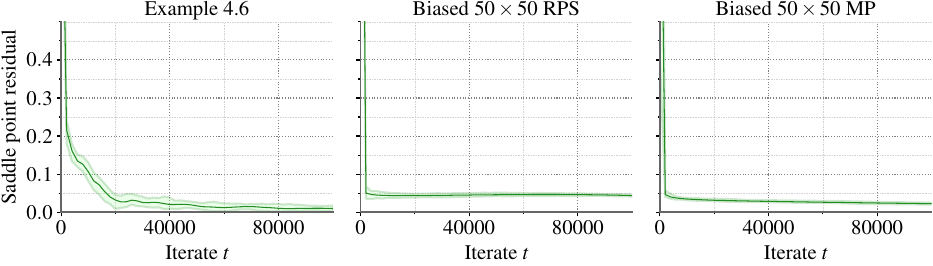}
    \caption{Saddle point residuals of the time-averaged empirical marginal pairs produced by CAT-FTPL for three 2p0s-GSAS. We repeated each experiment 100 times and plot the average and the 95\% central interval across the repetitions.}
    \label{fig:spr_plots_marginals_app}
\end{figure}

\begin{figure}[!ht]
    \centering
    \includegraphics[width=\linewidth]{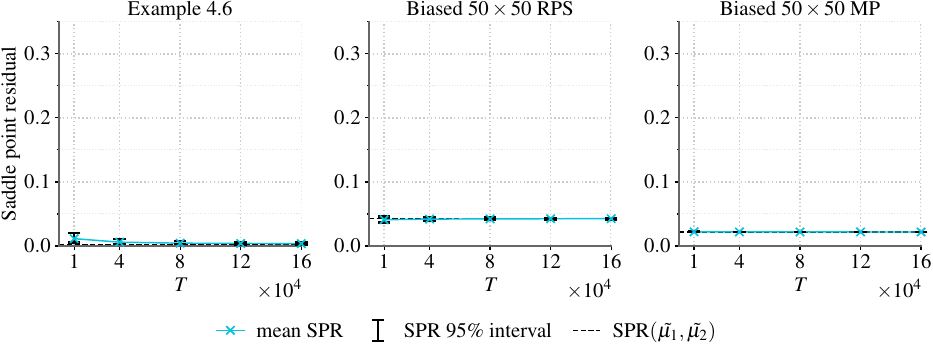}
    \caption{Saddle point residuals of the strategy pairs induced by the weight vectors \(w_i^T\) learned by~\cref{alg:compute_w_rsa} on the marginal produced by CAT-FTPL for three 2p0s-GSAS. We repeated each experiment 100 times and plot the mean and the 95\% central interval of the saddle point residual across the repetitions.}
    \label{fig:spr_plots_w_app}
\end{figure}

\subsubsection{Estimating the saddle point residual in experiments}\label{appsubsec:exp_spr_calculation}
Since we focus on the 2p0s
case, we can consider the \emph{saddle point residual} of a strategy pair $(\pi_1,\pi_2)$. In particular, let $U = U_1 = -U_2$. The SPR expressed in terms of strategies (rather than implementable marginals) is given as
\begin{definition}[Saddle point residual (SPR) of a strategy pair]
\label{def:spr_strategy}
    \begin{align*}
    \mathrm{SPR}(\pi_1,\pi_2) &= [U(\pi_1,\pi_{2}) - \min_{\pi'_2} U(\pi_1, \pi'_2)] + [\max_{\pi'_1} U(\pi'_1,\pi_{2}) - U(\pi_1, \pi_2)]\\
    &= \max_{\pi'_1} U(\pi'_1,\pi_{2}) - \min_{\pi'_2} U(\pi_1, \pi'_2).
\end{align*}
\end{definition}

It is easy to see that the saddle point residual of $(\pi_1,\pi_2)$ is equal to the saddle point residual of their marginals (cf.~\cref{def:SPR}) and is zero if and only if $(\pi_1,\pi_2)$ is a Nash equilibrium.

We seek to compute the saddle point residual of the strategies \(\pi_i\) both induced by the compact vectors from \cref{alg:compute_w} and by the marginal distributions produced by CAT-FTPL. 
However, \(U(\pi_1, \pi_2) = \mathbb{E}_{S \sim \rho}\left[\mathbb{E}_{a \sim \pi(S)}[u_1(a)]\right]\) is infeasible to calculate when we only have sampling access to \(\rho\), or when its support is too large to enumerate. We thus wish to calculate where possible, and estimate when not, the saddle point residual of the empirical marginals in three different regimes: (i) \(\rho\) is known and has small support, (ii) only sampling access to \(\rho\) is available, and (iii) when \(\rho\) is such that each action is available independently across the action set. For each regime, we describe the calculation only for \(\max_{\pi'_1} U(\pi'_1, \pi_2)\), noting the procedure for \(\min_{\pi'_2} U(\pi_1, \pi'_2)\) is similar. Although the empirical marginals may not be implementable, the chance that an empirical marginal is significantly different from an implementable marginal is negligible for the large number of time-steps used in our experiments.

\paragraph{(i) Known \(\rho\) with small support regime.}
If we know \(\rho = (\rho_1,\rho_2)\) we can directly calculate \(\max_{\pi'_1} U(\pi'_1, \pi_2)\). For a marginal distribution \(\mu_2\) induced by \(\pi_2\), we compute 
\begin{equation*}
\max_{\pi'_1} U(\pi'_1, \pi_2) = \sum_{S_1 \in \supp(\rho_1)} \mathbb{P}[S_1]\max_{a_1 \in S_1}\left[\sum_{a_2 \in A_2} \mu_2(a_2) U(a_1, a_2)\right].
\end{equation*}

If we instead have a strategy \(\pi_2\), we can first calculate the expected payoff for each action \(a_1 \in A_1\) as 
\begin{equation*}
    \mathbb{E}_{S_2 \sim\rho_2}[U(a_1; \pi_2(S_2))] = \sum_{S_2 \in\supp(\rho_2)}\rho_2(S_2)U(a_1;\pi_2(S_2)),
    \label{eqn:exp_computing_spr_action_expected_payoff}
\end{equation*}
and can then proceed the same as for the marginal. We use this procedure for calculating the SPR for \cref{ex:rps} and for small (\(n\le10\)) instances of ``Checkerboard MP'', ``Biased RPS'', and ``Biased MP''.

\paragraph{(ii) Sample \(\rho\) access regime.}
If we only have sample access to \(\rho\), for example when enumerating all possibly observed \(S_i \in\supp(\rho_i)\) is infeasible, we repeat the same process as in the known \(\rho\) regime but sampling from \(\rho\) as needed, instead of enumerating all possible action availabilities, to estimate the SPR.

We use this regime to estimate the SPR for large (\(n>10\)) instances of ``Checkerboard MP'', ``Biased RPS'', and ``Biased MP''.

\paragraph{(iii) Actions available independently regime.}
Consider a \(\rho_i\) such that action \(a_i\)'s availability is determined by independently sampling from a Bernoulli distribution with parameter \(p_{a_i}\), and resampling all actions if no action is available.
If we know the Bernoulli parameters \(p_{a_i}\), after estimating the expected payoff against the opponent \(E(a_i)\) for each action \(a_i \in A_i\) the same as in the ``sample \(\rho\) access regime'', we can then directly calculate
\(\mathbb{E}_{S_i \sim \rho_i}[\max _{a_i\in S_i} E(a_i)]\).

We proceed by calculating the best response for player 1, with the procedure identical for player 2. Assume without loss of generality that actions are indexed such that \(E(a_j) \geq E(a_{j+1})\) and let \(V=\max_{\pi'_1}U(\pi'_1, \pi_2) =\mathbb{E}_{S_1 \sim \rho_1} [\max_{a \in S_1}E(a)]\).

Ignoring the case that no action is present and \(S_1\) must be resampled, action \(a_i\) contributes
\(E(a_i)p_{a_i} \prod_{j=1}^{i-1} (1-p_{a_j})\)
to \(V\).

We account for the case where we need to resample by solving
\begin{align*}
    V &= \sum_{i=1}^{\lvert A_1 \rvert} \left[E(a_i)p_{a_i} \prod_{j=1}^{i-1} (1-p_{a_j})\right] + V\prod_{i=1}^{\lvert A_1 \rvert}(1-p_{a_i})\\
    &= \frac{1}{1-\prod_{i=1}^{\lvert A_1 \rvert}(1-p_{a_i})}\sum_{i=1}^{\lvert A_1 \rvert} \left[E(a_i)p_{a_i} \prod_{j=1}^{i-1} (1-p_{a_j})\right],
\end{align*}
which can be calculated in a single pass over the actions after sorting.

We use this to more accurately estimate the saddle point residual of a random 2p0s-GSAS (\cref{def:randomgsas}) and RBS GSAS (\cref{def:RBSgsas}).

\subsubsection{Calculating the ranking regret in experiments}\label{appsubsec:exp_regret_calculation}
Naively calculating the ranking regret directly by the definition (cf.~\cref{def:ranking-regret}), i.e., by tracking the realized payoff for every ranking strategy \(\pi_i^{\tau_i}\), 
would require tracking \(\lvert A_i \rvert!\) many values (the payoff from each ranking) which is infeasible for large GSAS. Thus, to compute \(R^\mathrm{rank}_{T,i}\) in practice, we:
\begin{enumerate}
    \item Calculate the expected payoff from each action given the opponents' marginals using~\cref{eqn:exp_computing_spr_action_expected_payoff}.
    \item Calculate the optimal ranking \(\pi_i^{\tau_i}\) in hindsight.
    \item Estimate \(\sum_{t=1}^T U_i(\pi_i^{\tau_i}, \pi_{-i}^t)\) using the same three calculation regimes as estimating the SPR (the ``known \(\rho\) with small support'' regime, the ``sample \(\rho\) access'' regime or the ``actions available independently'' regime).
    \item Subtract \(\sum_{t=1}^T U_i(\pi^t)\).
\end{enumerate}
Note that similar to estimating the saddle point residual, the realized payoff may differ from the expected payoff of the strategy, but this difference is negligible for the large values of \(T\) used in our experiments.

\subsection{Solving Example~\ref{ex:rps} Computationally}\label{appsec:exampleRSP}
\begin{definition}\label{def:rps-mod}
    The original payoff matrix $\calG_\mathrm{orig}$ of the modified RPS game from~\cref{ex:rps} is given by
\begin{align*}
        \begin{array}
        {c@{\hspace{1em}}ccc}
        \toprule
         & \textbf{R}\text{ock}& \textbf{P}\text{aper} & \textbf{S}\text{cissors}  \\
        \midrule
        \textbf{R}\text{ock}     & 0  & -1 & 1 \\
        \textbf{P}\text{aper} & 0.5 & 0  & -0.5 \\\textbf{S}\text{cissors}    & -1  & 1  & 0 \\
        \bottomrule
        \end{array}
    \end{align*}
\end{definition}



We apply the procedure outlined in~\cref{sec:compute} to solve the GSAS as described in~\cref{ex:rps}. In particular, we run CAT-FTPL and run~\cref{alg:compute_w_rsa}. We show the learned weights per iterate of~\cref{alg:compute_w_rsa}  in~\cref{fig:examplegame}. We find that the weights learned by \cref{alg:compute_w_rsa} converge to the analytically computed compact Nash equilibrium vector, which is $(\frac{1}{2},\frac{1}{3},\frac{1}{6})$ for player 1 and $(\frac{2}{5},\frac{2}{5},\frac{1}{5})$ for player 2.


\begin{figure}[ht]
    \centering
    \includegraphics[width=0.6\linewidth]{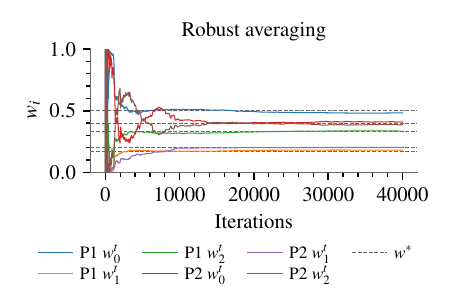}
    \caption{Learned weights by \cref{alg:compute_w_rsa} using robust averaging on CAT-FTPL output in~\cref{ex:rps}.}
    \label{fig:examplegame}
\end{figure}

\subsection{Comparison with Multi-scale MWU for Bayesian Games}
\label{appsubsec:multiscalemwu}
In \cref{fig:multiscale-MWU}, we show the wallclock running time for multi-scale MWU~\citep{peng2024complexity} as applied to random 2p0s-GSAS (as defined in~\cref{def:randomgsas} and used in Experiment 1) with \(\epsilon=0.3\).
Note that the choice of $\epsilon=0.3$ for multi-scale MWU was a practical one: the theory guarantees convergence to an $\epsilon$-normal form correlated equilibrium of the expanded game, but we empirically observe that the runtime of multi-scale MWU was prohibitively slow for smaller $\epsilon$. As the expanded game is two-player zero-sum, normal form correlated equilibria are equivalent to Nash equilibria.

Recall that multi-scale MWU was developed for Bayesian games~\citep{peng2024complexity,dagan2024external}, and one can reduce a GSAS of size $n$ to a Bayesian game with \(2^n - 1\) types per player (one type for each nonempty availability set). In \cref{fig:multiscale-MWU}, we observe the runtime grows exponentially with game size for multi-scale MWU, similar to the LP solver method in \cref{fig:runtime_experiment}. This observation matches the bounds given in~\citep{peng2024complexity}: although the bound on the iterations needed is independent of the number of types in the Bayesian game, the bound on the total runtime grows in the number of types.

CAT-FTPL in comparison exhibits runtime polynomial in the game size, matching the theoretical bounds. This difference is expected, as CAT-FTPL can exploit the additional structure in GSAS between availability sets, which is inaccessible to a Bayesian game solver where there may not be structure across types.


\begin{figure}[ht]
    \centering
    \includegraphics[width=0.6\linewidth]{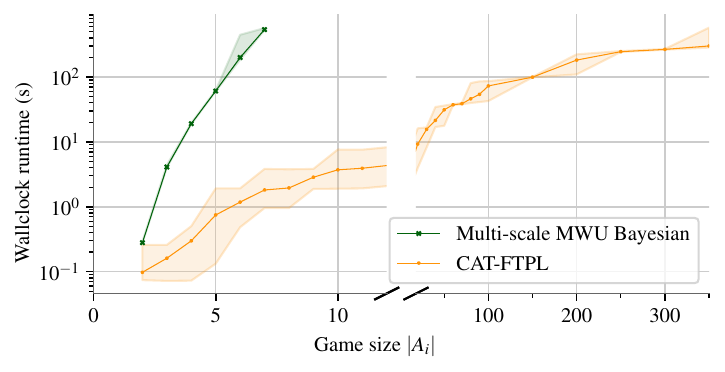}
    \caption{Wallclock runtime for multi-scale MWU on random GSAS (\cref{def:randomgsas}). In a time budget of 360 s the algorithm can solve a $6 \times 6$ game, with 609 s on average needed to solve a $7 \times 7$ random GSAS. Comparative results from \cref{fig:runtime_experiment} for our method are shown, scaling to $350 \times 350$ random GSAS within the 360 s time budget. }
    \label{fig:multiscale-MWU}
\end{figure}



\section{On the Bit Complexity of GSAS Solutions}\label{appsec:bitcomplexity}
Before proceeding, we make two important points.
First, we are \textit{not} interested in bit representations that are as large as the number of entries of $\rho$ (which could be exponential in \(\lvert A_i \rvert\)) --- if we allowed that, then one may as well do away with the compact representation altogether and work with the naïve distribution mapping action sets to distributions over actions, an approach that is neither practical nor enlightening.
Second, this section is only concerned with \textit{representing} (approximate) equilibria, and has little to do with explicit algorithms to compute the representations.


\subsection{Bit-Compact Representations}
Fix some family \(\Gamma\) of 2p0s-GSAS, for example, those meeting~\cref{assumption:product}.
We want to be able to represent an \(\epsilon\)-NE strategy for any GSAS from \(\Gamma\). Informally, this representation should allow us to play an \(\epsilon\)-NE when given an availability set, without needing to perform further game-specific calculations.\footnote{For example, we cannot allow the naïve strategy of solving the GSAS from scratch when presented with an availability set.} As a concrete example, for \(\Gamma = \{\mathcal{G} \mid \mathcal{G} \text{ a 2p0s-GSAS meeting~\cref{assumption:product}}\}\), we know from~\cref{thm:compact} that for every \(\mathcal{G} \in \Gamma\) there exists a vector \(w_i \in \Delta(A_i)\) and the following algorithm: given \(w_i\) and an availability set \(S_i\), renormalize \(w_i\) over the available actions, will play an \(\epsilon\)-NE for \(\mathcal{G}\). Note that to strictly meet the definition of bit-compact representation would require showing we can approximate the numbers in \(w_i\) with polynomially many bits.

While we do not want to unnecessarily restrict how such bit-compact  \(\epsilon\)-NE representations work, we do require that they admit a \emph{playing algorithm} which is independent of any particular GSAS in \(\Gamma\). For any GSAS in \(\Gamma\), given the right bit string and an availability set, the \emph{playing algorithm} should play a valid \(\epsilon\)-NE. To this end,~\cref{def:compactrep} formally captures `bit-compactness' of $\epsilon$-NE in 2p0s-GSAS:

\begin{definition}[Bit-compact representable \(\epsilon\)-NE in 2p0s-GSAS]
	\label{def:compactrep}
    Let \(\Gamma \subseteq \{\mathcal{G} \mid \mathcal{G} \text{ is a 2p0s-GSAS}\}\) be a family of GSAS, with \(\mathcal{G} \in \Gamma\) a particular GSAS from the family.
    We denote \(\mathcal{G}\)'s actions for player \(i\) by \(A_{\mathcal{G}, i}\) and availability sets by \(\mathcal{S}_{\mathcal{G}, i}\).
	
We say that \(\Gamma\) has \emph{bit-compactly representable \(\epsilon\)-NE} if there exists a polynomial \(p\) and an algorithm \(R\), taking as input a bit string \(x_{\mathcal{G},i}\) and an availability set \(\calS_{\mathcal{G},i}\), such that \(\forall \epsilon\in(0,1), \forall \mathcal{G} \in \Gamma, \forall i \in \{1,2\},\exists x_{\mathcal{G},i} \in \{0,1\}^{p\left(\lvert A_{\mathcal{G},i} \rvert, \log\left(\frac{1}{\epsilon}\right)\right)}\), $\forall S_i \in \mathcal{S}_{\mathcal{G},i}$, $R(S_i, x_{\mathcal{G},i}) \in \Delta(S_i)$ and the strategy profile played by R is an $\epsilon$-NE. 

We call such an \(x_{\mathcal{G},i}\) for a GSAS \(\mathcal{G}\) a \emph{bit-compact representation} and let \(k := \lvert x_{\mathcal{G},i} \rvert\) be the length of the bit string.
\end{definition}

Crucially, \(R\) is fixed for \(\Gamma\) and only takes in \(x_{\mathcal{G}, i}\) and \(S_i\), but not the entire GSAS \(\mathcal{G}\). This means that it can only operate on the compact representation afforded by \(x_{\mathcal{G}, i}\), and cannot rely on external quantities like \(\rho\) unless they are stored either implicitly or explicitly in \(x_{\mathcal{G}, i}\). Furthermore, note that \(k\) is required to be polynomial in \(\lvert A_i \rvert\) and \(\log\left( \frac{1}{\epsilon} \right)\), while in general \(\rho\) could require \(\Omega(2^{\lvert A_i \rvert})\) bits to encode. The omission of \(\mathcal{G}\) was explicitly chosen to ensure any bit-compact representation contains all the information needed to play an \(\epsilon\)-NE, but without specifying how that information is represented.

Note that having a bit-compact representation of a GSAS family is necessary, but not sufficient, to have a practical method to solve that family: one additionally needs a practical method to compute the bit string, and for the playing algorithm to run in polynomial time in its input length. However, our procedure in~\cref{sec:compute} meets these additional requirements, with CAT-FTPL and~\cref{alg:compute_w_rsa} providing a practical method to compute \(w\) and the playing algorithm used by~\cref{thm:compact} having a polynomial runtime.

Finally, while~\cref{def:compactrep} is specified for 2p0s-GSAS, it can be broadened to general GSAS by ensuring that the joint strategy played by \(R\) is an \(\epsilon\)-NE. The definition can also be easily generalized to cover solution concepts such as (coarse) correlated equilibria, or indeed any set of strategies meeting some desired criteria.

\subsection{Polynomial Bit-Compactness under Assumption~\ref{assumption:product}}\label{appsec:counterexample}

In light of the above, a nuanced but important question is: while the NE of a 2p0s-GSAS meeting~\cref{assumption:product} can be approximately represented for player $i$ with $O(|A_i|)$ real numbers given by $w^*_i$, can the entries of $w^*_i$ be written in a polynomial number of \textit{bits} that are used to represent the original payoff matrix $\mathcal{G}_\mathrm{orig}$ (whose elements are represented by rational numbers)? 
If this were indeed the case, then such a $w^*_i$ can be used as a bitstring for $x_{\mathcal{G},i}$, and that would immediately be a bit-compact representation in the sense of~\cref{def:compactrep}.

Unfortunately, the answer to the above is negative. Indeed, the proof of~\cref{thm:compact} relies on $w^*_i$ being the solution to a system of quadratics (cf.~\cref{eqn:quadraticsystem}). It is well-known that solutions to quadratic systems can be algebraic, even if the entries of the payoff matrices are rational. 
Concretely, we give a counterexample showing that there exists a 2p0s-GSAS $\mathcal{G}$ 
for which a compactly
represented equilibrium $w^*$ (in the sense of Theorem~\ref{thm:compact}) contains irrational numbers as its entries, even under \cref{assumption:product}, assuming payoffs in $\mathcal{G}_\mathrm{orig}$ and probabilities in $\rho$ are all rational. 

Consider the following 2p0s-GSAS (with independent action availabilities, fulfilling \cref{assumption:product}) where only player 1 has stochastic action sets. We have $|A_1|=|A_2|=3$ and 
    \begin{align*}
        \mathcal{G}_\mathrm{orig} = 
        \begin{array}
        {c|@{\hspace{1em}}ccc}
        A_1 \backslash A_2 & 1 & 2 & 3  \\
        \hline
        1     & 1/2  & 0 & 0 \\
        2 & 0 & 1/2  & 0 \\
        3 & 0  & 0  & 1/3 
        \end{array}
    \end{align*}
    where $\mathcal{S}$ and $\rho$ are given by 
    \[
    \renewcommand{\arraystretch}{1.4}
    \begin{array}{c|c}
    \mathcal{S}_1 & \rho_1 \\
    \hline
    \{1\}       &\frac{1}{7} \\
    \{2\}       &\frac{1}{7} \\
    \{3\}       &\frac{1}{7} \\
    \{1,2\}     &\frac{1}{7} \\
    \{1,3\}     &\frac{1}{7} \\
    \{2,3\}     &\frac{1}{7} \\
    \{1,2,3\}   & \frac{1}{7} \\
    \end{array}, \quad
    \begin{array}{c|c}
    \mathcal{S}_2 & \rho_2 \\
    \hline
    \{1,2,3\}       & 1\\
    \end{array},
    \]
    i.e., 
    player 1’s availability set is sampled uniformly from the nonempty subsets of $A_1$,
    while player 2 has every action always available. 
    It is easy to verify that the unique equilibrium of $\mathcal{G}_\mathrm{orig}$ is given by both players playing actions 1, 2, and 3 with probabilities $[2/7, 2/7, 3/7]$ respectively. We know from \cref{thm:gsas-contains-ne} that if this distribution is implementable by both players, then it is also an equilibrium in \(\mathcal{G}\) in a sense of marginal distributions as described by $(\mu_1^*, \mu_2^*)$. Indeed, this can be achieved in several ways. For example, the following gives the joint probabilities for an equilibrium strategy of player 1:
    \[
    \renewcommand{\arraystretch}{1.4}
    \begin{array}{c|ccc|c}
    \mathcal{S}_1 & 1 & 2 & 3 & \rho_1 \\
    \hline
    \{1\}       & \frac{1}{7}  & 0             & 0             & \frac{1}{7} \\
    \{2\}       & 0            & \frac{1}{7}   & 0             & \frac{1}{7} \\
    \{3\}       & 0            & 0             & \frac{1}{7}   & \frac{1}{7} \\
    \{1,2\}     & \frac{1}{14} & \frac{1}{14}  & 0             & \frac{1}{7} \\
    \{1,3\}     & \frac{1}{28} & 0             & \frac{3}{28}  & \frac{1}{7} \\
    \{2,3\}     & 0            & \frac{1}{28}  & \frac{3}{28}  & \frac{1}{7} \\
    \{1,2,3\}   & \frac{1}{28} & \frac{1}{28}  & \frac{1}{14}  & \frac{1}{7} \\
    \hline
    \mu_1^*
                 & \frac{2}{7} & \frac{2}{7}   & \frac{3}{7}   & 1 \\
    \end{array},
    \]
    which unfortunately does not obey our compact representation. However, the following implements the exact same marginal probabilities $\mu_1^*$:
    \[
    \renewcommand{\arraystretch}{1.4}
    \begin{array}{c|ccc|c}
    \mathcal{S}_1 & 1 & 2 & 3 & \rho_1 \\
    \hline
    \{1\}       & \frac{1}{7}        & 0                    & 0                    & \frac{1}{7} \\
    \{2\}       & 0                    & \frac{1}{7}        & 0                    & \frac{1}{7} \\
    \{3\}       & 0                    & 0                    & \frac{1}{7}        & \frac{1}{7} \\
    \{1,2\}     & \frac{1}{14}       & \frac{1}{14}       & 0                    & \frac{1}{7} \\
    \{1,3\}     & \frac{s-3}{28}     & 0                    & \frac{7-s}{28}     & \frac{1}{7} \\
    \{2,3\}     & 0                    & \frac{s-3}{28}     & \frac{7-s}{28}     & \frac{1}{7} \\
    \{1,2,3\}   & \frac{5-s}{28}     & \frac{5-s}{28}     & \frac{s-3}{14}     & \frac{1}{7} \\
    \hline
    \mu_1^*
                 & \frac{2}{7}        & \frac{2}{7}        & \frac{3}{7}        & 1 \\
    \end{array},
    \]
    where $s=\sqrt{17}$ does indeed obey our compact representation given by 
    \begin{align*}
        w^*_1 = \left[\frac{5-s}{4}, \frac{5-s}{4}, \frac{s-3}{2}\right],
    \end{align*}
    which is clearly irrational. 

The above counterexample may look discouraging, 
as it shows that the normalized proportional weights need not be rational, so the decimal/binary representation cannot be written exactly with finite bits.
However, 
we show that there exist $\epsilon$-NE compactly represented by $\tilde{w}^*$ that do enjoy polynomial bit-complexity. This can be achieved easily by a simple dyadic truncation procedure. 

For a positive weight vector $w$, we require the following notation:
\begin{equation}\label{eq:marginaldefbit}
    w(S)=\sum_{a\in S}w_a,
    \qquad
    \hat\mu(w)_a
      =\sum_{S\ni a}\rho(S)\frac{w_a}{w(S)}.
\end{equation}

In other words, $\hat\mu(w)$ is the marginal induced by playing the GSAS using $w$. 
Then, we have the following result:
\begin{theorem}\label{thm:approxbitcompactexist}
    For every implementable marginal $\mu_i\in M_i$ and each $\delta\in(0,1)$, there exists a $\delta$-approximate compact representation $w_i$ of $\mu_i$ in the sense of Theorem~\ref{thm:compact} and a rational vector $\tilde w_i \in \Delta(A_i)$ that can be represented in $O(|A_i|\log(|A_i|/\delta))$ bits such that $\|\tilde w_i - w_i\|_1 \leq \delta$ and $\| \hat{\mu}_i(\tilde w_i) - \mu_i\|_1 \leq \delta$.
\end{theorem}
\begin{proof}
 In the remainder of the proof, we suppress the player index for brevity. We will require a technical lemma that establishes an appropriate choice of $w$ so that the number of bits required to represent $\tilde w$ 
 is $O(|A|\log(|A|/\delta))$.
 \begin{lemma}\label{lem:technicalbitcompact}
Let $\mu\in M$. For every $\eta\in(0,1)$, there exists $w\in\Delta(A)\cap\mathbb R_{>0}^{|A|}$ such that
\begin{equation}\label{eq:technicallemmaconditions}
    \|\hat\mu(w)-\mu\|_1\le 2\eta,
    \quad \text{and} \quad
    \min_a w_a
      \ge \frac{1}{|A|}\left(\frac{\eta}{|A|}\right)^{|A|-1}.
\end{equation}
\end{lemma}

\begin{proof}
     Consider a $\pi$ that implements $\mu$. For $a\in S$, and similarly to the proof of Theorem~\ref{thm:compact} we define the conditional distribution:
    \begin{equation}
        \pi^\eta(a | S) = (1-\eta) \pi(a | S) + \frac{\eta}{|S|},
    \end{equation}
    which ensures full support for $\pi^\eta$ for all $S$. 
    Let $\mu^\eta$ be the induced marginal of $\pi^\eta$. Taking the expectation over action availabilities with respect to $\rho$ gives:
    \begin{equation}
        \|\mu^\eta - \mu\|_1 \leq \sum_S \rho(S) \| \pi^\eta(\cdot|S)-\pi(\cdot|S)\|_1 \leq 2\eta,
    \end{equation}
    where the first inequality holds due to the triangle inequality and the second holds because the $\ell_1$-norm between two probability distributions is at most $2$.
    Then, the same matrix-scaling argument in the proof of Theorem~\ref{thm:compact} ensures that there exist positive weights $v$ such that $\hat\mu(v) = \mu^\eta$.

    Next, we seek to modify the weights $v$ without changing the induced marginal, so as to ensure a lower bound on the smallest element of $w$. First, let us relabel the actions so that $v_1\geq \dots\geq v_{|A|}>0$. Moreover, for each $j<|A|$, consider the sets $H_j = \{1,\dots,j\}$ and the ratios $r_j = \frac{v_{j+1}}{v_j}$. Next, we show that for any admissible availability set, one can rescale the weights without changing the proportional probabilities within the availability set.

    Suppose $r_j < \eta/|A|$. Then, for any availability set $S$ that satisfies $S\cap H_j \neq \varnothing$ and $S \setminus H_j \neq \varnothing$ (i.e., an availability set that contains at least one action in $H_j$ and one outside of $H_j)$, it holds that:
    \begin{align*}
        \sum_{a\in S\setminus H_j} \frac{v_a}{v(S)} \leq \frac{|S\setminus H_j|}{|S\cap H_j|} \cdot r_j
        < \frac{\eta|S\setminus H_j|}{|S|}
        \leq \sum_{a\in S\setminus H_j} \pi^\eta(a|S).
    \end{align*}
    
    The strict inequality above holds because $|S| \leq |A|\cdot|S\cap H_j|$. Hence, whenever an availability set contains actions both in and out of $H_j$, $\pi^\eta$ assigns strictly more total probability to actions outside $H_j$ than the normalized weights do. 
    In particular, it holds that $\sum_{a\in S\setminus H_j} \left(\pi^\eta(a|S) - \frac{v_a}{v(S)}\right) \geq 0$ for every availability set in the game. This implies that if such an availability set occurs with strictly positive probability, the induced marginals satisfy $\sum_{a\notin H_j}(\mu_a^\eta - \hat\mu_a(v))=\sum_S\rho(S)\sum_{a\in S\setminus H_j}(\pi^\eta(a\mid S)-\frac{v_a}{v(S)}) > 0$. This contradicts equality of the induced marginals as established above. Hence, there does not exist any availability set with positive probability that contains actions both in and out of $H_j$. 

    The preceding argument ensures that multiplying $v_{j+1},\dots,v_{|A|}$ (i.e., the weights outside of $H_j)$ by some constant $c$ does \emph{not} change the strategy induced by the rescaled compact vector. In our setting, we can select $c = \frac{\eta}{|A| r_j}$, which ensures that the ratio after rescaling is $r'_j = \frac{v'_{j+1}}{v'_j} = \frac{\eta}{|A|}$. This also preserves the ordering and leaves all other consecutive ratios unchanged. Hence, we can perform this rescaling for each $j$ with $r_j < \frac{\eta}{|A|}$, then normalize the resulting vector to obtain a suitable $w$. By construction, $\hat\mu(w) = \mu^\eta$, and $\frac{w_{j+1}}{w_j} \geq \frac{\eta}{|A|}$ for each $j<|A|$. Since we have that $w_1 \geq 1/|A|$, it holds that
    \begin{equation}
        \min_a w_a = w_{|A|} = w_1\prod_{j=1}^{|A|-1} \frac{w_{j+1}}{w_j} \geq \frac{1}{|A|} \left(\frac{\eta}{|A|}\right)^{|A|-1}.
    \end{equation}

    This concludes the proof of the lemma.
\end{proof}
 
 In the remainder of the proof, we apply Lemma~\ref{lem:technicalbitcompact} with $\eta = \frac{\delta}{4}$ to obtain a $w\in \Delta(A)$ that satisfies the conditions in~\eqref{eq:technicallemmaconditions}. We consider an ordering of the elements of $w$ so that $w_1 = \max_a w_a$, where $w_a$ is the $a$-th element of $w$ and $1\leq a\leq |A|$. We first show that $\|\tilde w - w\|_1 \leq \delta$.
    Select $b = \lceil \log_2 (\frac{4|A|}{\delta})\rceil$ and $e_a = \lceil 
    \log_2\frac{1}{w_a}\rceil$. Then, we can define $q_a$ to be $w_a$ truncated downward, i.e., 
    \begin{equation}\label{eq:dyadictruncation}
        q_a = \frac{\lfloor 2^{b+e_a} w_a \rfloor}{2^{b+e_a}}.
    \end{equation}
    
    Since $2^{-e_a} \leq w_a < 2^{1-e_a}$, it holds that $2^{b} \leq \lfloor 2^{b+e_a} w_a \rfloor < 2^{b+1}$. Thus, $0\leq w_a - q_a < 2^{-b-e_a} \le 2^{-b}w_a \leq 2^{-b}$. Hence, each $q_a$ is positive and rational.

    
    However, $q$ may sum to below $1$. Let the deficit be $\kappa = 1-||q||_1 \geq 0$. To fix this deficit in our approximate $\tilde{w}$, we can add $\kappa$ to $q_1$ corresponding to $w_1 = \max_a w_a$ such that 
    \begin{align*}
        \tilde{w}_a = \begin{cases}
            q_1 + \kappa \qquad & a = 1 \\
            q_a \qquad & a \neq 1 
        \end{cases}.
    \end{align*}
    
    Clearly, $\tilde{w}$ is rational since $\kappa$ is also rational.
    Then, we have
    \begin{align*}
        ||\tilde{w} - w||_1 &= 
        |w_1 - \tilde{w}_1| + 
        \sum_{a \neq 1} |w_a - \tilde{w}_a| \\
        &\leq \kappa + ||q - w||_1 \\
        &\leq |A| \cdot 2^{-b} + |A| \cdot 2^{-b} \\
        &= 2|A|\cdot 2^{-b} \\
        &\leq \delta.
    \end{align*}

Note that despite the above, it remains to control the number of bits needed to represent $\tilde w$. In particular, notice that $e_a$ could require arbitrarily many bits without a lower bound on $w_a$. 

Using Lemma~\ref{lem:technicalbitcompact}, we proceed to show the number of bits representing $\tilde w$. In particular, by selecting $\delta = 4\eta$, the lemma ensures that $\max_a e_a \leq \left\lceil \log_2 |A| + (|A|-1) \log_2\frac{4|A|}{\delta}\right\rceil \leq |A| \cdot b$. Each truncation in Equation~\eqref{eq:dyadictruncation} requires $b+1$ bits, each $e_a$ requires $O(\log(|A| \cdot b + 2))$ bits, and $\tilde w_1=q_1+\kappa$ also requires $O(|A| b)$ bits as its denominator divides $2^{b+\max_a e_a}$. Thus, overall, we require $O(|A| (b + \log(|A| b + 2))) = O(|A| b) = O(|A| \log(|A|/\delta))$ bits. This is because we have $b = O(\log(|A|/\delta))$ and $\log |A| \leq b$.

Next, we need to establish that the induced marginal $\hat\mu(\tilde w)$ approximates $\mu$, i.e., $\|\hat\mu(\tilde w) - \mu\|_1 \leq \delta$. First, by summing over actions, we get that
\begin{equation}
    \kappa = \sum_{a}(w_a-q_a)\leq 2^{-b}.
\end{equation}

For $a\neq 1$, we have that $|\tilde w_a - w_a| \leq 2^{-b} w_a$. For action 1, since $w_1 \geq 1/|A_i|$, we have
$0 \le \tilde w_1 - w_1 = \sum_{a\neq 1} (w_a-q_a) \leq 2^{-b} \leq |A| 2^{-b} w_1$.

Thus, for each $a$, it holds that $|\tilde w_a - w_a|\leq |A| 2^{-b} w_a$. Towards obtaining the marginals corresponding to $\tilde w$ and $w$, using the notation established in Equation~\eqref{eq:marginaldefbit}, we have:
\begin{align*}
    \sum_{a\in S} \left| \frac{\tilde w_a}{\tilde w(S)} -\frac{ w_a}{ w(S)} \right| &\leq \sum_{a\in S} \tilde w_a\left| \frac{1}{\tilde w(S)} -\frac{1}{ w(S)} \right| + \frac{\sum_{a\in S}|\tilde w_a - w_a|}{w(S)}\\
    &= \frac{|w(S) - \tilde w(S)| + \sum_{a\in S}|\tilde w_a - w_a|}{w(S)}\\
    &\leq  \frac{2\sum_{a\in S}|\tilde w_a - w_a|}{w(S)}\\
    &\leq 2 |A| 2^{-b} \leq \frac{\delta}{2}
\end{align*}

The penultimate inequality uses the fact that $|\tilde w_a - w_a|\leq |A| 2^{-b} w_a$. Finally, taking the empirical average over availability sets according to $\rho$ gives
\begin{equation}
    \|\hat\mu(\tilde w) - \hat\mu(w)\|_1 \leq \delta/2.
\end{equation}

From Lemma~\ref{lem:technicalbitcompact}, we also have that $\|\hat\mu(w) - \mu\|_1 \leq 2\eta = \delta/2$. Thus,  
\begin{align}
    \|\hat\mu(\tilde w) - \mu\|_1 \leq \|\hat\mu(\tilde w) - \hat\mu(w)\|_1 + \|\hat\mu(w) - \mu\|_1 \leq \delta,
\end{align}
as required.
\end{proof}

As a corollary, it follows that for any 2p0s-GSAS that meets Assumption~\ref{assumption:product}, the dyadic truncation procedure above ensures the existence of an \emph{approximate} bit-compact representation of NE.

\begin{cor}\label{cor:approxbitcompact}
    Every 2p0s-GSAS that satisfies Assumption~\ref{assumption:product} admits, for every $\epsilon\in (0,1)$, an $\epsilon$-NE which, for player $i$, can be represented by $O(|A_i| \log (|A_i|/\epsilon))$ bits.
\end{cor}
\begin{proof}
    Let $(\mu_1^*,\mu_2^*)$ be marginals corresponding to NE in a 2p0s-GSAS meeting the assumption. We can apply Theorem~\ref{thm:approxbitcompactexist} to each $\mu_i^*$ with $\delta = \epsilon/2$, which ensures $\|\hat \mu_i(\tilde w_i) - \mu^*_i\|_1 \leq \epsilon/2$. Then, applying the SPR inequality from the proof of Theorem~\ref{thm:mainthm} (Equation~\eqref{eq:sprbound}), we get
    \begin{align*}
        \mathrm{SPR} (\hat\mu_1(\tilde w_1), \hat\mu_2(\tilde w_2)) &\leq \mathrm{SPR} (\mu^*_1, \mu^*_2) + \|\hat\mu_1(\tilde w_1)-\mu_1^*\|_1 + \|\hat\mu_2(\tilde w_2) - \mu^*_2\|_1\leq \epsilon,
    \end{align*}
    as required.
\end{proof}

\subsection{Non-Existence of Bit-Compact Nash Equilibria in 2p0s-GSAS}
\label{appsec:nocompact}


In this section, we provide a result that establishes non-existence of bit-compact approximate NE strategies in 2p0s-GSAS that do not meet~\cref{assumption:product}.

\begin{theorem}[No \(\epsilon\)-NE bit-compact representation]
	\label{thm:nocompact}
	The set
    \[\Gamma = \{\mathcal{G} \mid \mathcal{G} \text{ is a 2p0s-GSAS}\}\]
    does not have a bit-compact representation of \(\epsilon\)-NE.
\end{theorem}



 
 \cref{thm:nocompact} further justifies the use of~\cref{assumption:product} in our analysis: indeed, once the assumption is relaxed, there does not exist an algorithm that can play a bit-compact representation of $\epsilon$-NE in the family of 2p0s-GSAS.

\paragraph{Proof sketch.} Our proof of \cref{thm:nocompact} is conceptually similar to counting arguments, such as those in \cite[Chapter 6]{Arora_Barak_2009}, to show circuit lower bounds: we construct a doubly exponentially large family of games, all of which must have different bit-compact representations in order to play an \(\epsilon\)-NE.
The family of interest arises from a mapping from a binary function \(f\) to a 2p0s-GSAS, with a matching pennies setup and \(m\) additional `signal' actions for player 1 which are strictly dominated by Heads. The availability sets are such that player 1 always has Heads and Tails available, while the availability of the signal actions is used as input to \(f\), with player 2 having exactly one of Heads or Tails available, as determined by \(f\). Since player 2 can never deviate from their single valid strategy, a Nash Equilibrium is for player 1 to use the availability of its signal actions to also play Heads or Tails according to \(f\). 

We use the Gilbert-Varshamov bound~\citep{gilbert1952comparison,varshamov1957estimate} from coding theory, which is written explicitly in~\cref{thm:gilvert-varshamov} following~\cite[Section 4.2]{guruswami2025essential} and~\cite{druk2014lineartime}. We use this bound 
to construct an exponentially large family of games where for all games in the family, and for a constant fraction of the sampled signal availabilities, player 2 receives a different available action. We use a counting argument to establish that if there exists an \(\epsilon\)-NE bit-compact representation, when \(\epsilon < \frac{1}{4}\) for at least two games in the family player 1 must receive the same bit string and must play the same way on both games. 
This collision allows us to establish a contradiction, since by construction the same strategy cannot be an \(\epsilon\)-NE for both games.

\begin{theorem}[Gilbert-Varshamov Bound for \(q=2\)]
\label{thm:gilvert-varshamov}
For every \(0 < \delta < \frac{1}{2}\) there exists a family of binary linear codes \(\mathcal{C}\) with rate \(R(\mathcal{C}) \geq 1 - H_2(\delta)\) and relative distance \(\delta(\mathcal{C}) \geq \delta.\)
\end{theorem}

From the Gilbert-Varshamov bound there exists a set \(C\) of binary strings of length \(l\) with
\begin{equation}
\label{eqn:gilbert}
    \lvert C \rvert \geq \frac{2^l}{\text{Vol}_2(\delta l, l)} \geq 2^{l(1-H_2(\delta))} \quad \text{ and } \min_{\substack{c_1, c_2 \in C \\ c_1 \neq c_2}}  d_{\text{Hamming}}(c_1, c_2) \geq \delta l,
\end{equation}
where \(\text{Vol}_2(\delta l, l)\) is the volume of a Hamming ball of radius \(\delta l\) in \(\{0,1\}^l\) and \(H_2(\delta) = -\delta \log_2 (\delta) - (1-\delta)\log_2(1-\delta)\) is the binary entropy function.

\begin{proof}[Proof of~\cref{thm:nocompact}]
\textbf{Game generator definition.} We first define \(\mathcal{G}(f)\), which constructs a 2p0s-GSAS from any binary function \(f: \{0,1\}^m \to \{0,1\}\).

Let \(A_1 = \{H,T,1,2,\dots,m\}\) and \(A_2 = \{H,T\}\) be the action sets of players 1 and 2 respectively. Moreover, let the (zero-sum) payoff function be
\begin{align}
    u_1(a_1,a_2) = -u_2(a_1,a_2) = \begin{cases} 1 & \text{ if } a_1 = a_2, \\ 0 & \text{ if } a_1 \neq a_2, a_1 \in \{H,T\}, \\ -1 & \text{ if }a_1 \in \{1,2,\dots,m\}.\end{cases}
\end{align}

Let the game have availability sets \(\mathcal{S}_1 = \{S_{1,j}\mid j \subseteq \{1,2,\dots,m\}\}\) with $S_{1,j} = \{H,T\} \cup j$,
and 
\(\mathcal{S}_2 = \{S_{2,H},S_{2,T}\}\) with \(S_{2,H} = \{H\}, S_{2,T} = \{T\}\). The distribution over availability sets is
\begin{align}
    \rho(S_{1,j}, S_2) = \begin{cases}
		\frac{1}{2^m} & \text{ if } f(\bbone_j) = 1, S_2 = \{H\},\\
		\frac{1}{2^m} & \text{ if } f(\bbone_j) = 0, S_2 = \{T\},\\
		0 & \text{ otherwise},\end{cases}\label{eqn:nocompact_rho}
\end{align}
where $\bbone_j \in \{0,1\}^m \text { with } \bbone_j(i) \coloneqq \begin{cases} 1 & i \in j \\ 0 & i \notin j\end{cases}$ is the Boolean indicator vector.

\textbf{Game family construction.}
With this method to generate a 2p0s-GSAS from a binary function \(f\), we now construct a suitable family of games.
Assume toward a contradiction that there exists an \(\epsilon\)-NE bit-compact representation \(R\). Fix an \(\epsilon < \frac{1}{4}\) and let \(\delta\) be a constant such that \(\frac{1}{2} > \delta > 2\epsilon\).
From the Gilbert-Varshamov bound \eqref{eqn:gilbert}, we know there exists a set of binary strings \(C \subseteq \{0,1\}^{2^m}\) with relative distance \(\delta\) and \(\lvert C \rvert \geq 2^{(1-H_2(\delta))2^m}\). Therefore, for any \(R\), with \(k \in \mathrm{poly}\left(m+2, \log\left(\frac1\epsilon\right)\right) \) the length of the bit-compact representation bit string for a game, there exists an \(m_0\) such that \(\forall m \geq m_0, 2^{2^{m(1-H_2(\delta))}} > 2^k\).

Fix such an \(m \geq m_0\) and let \(C\) be such a set of binary strings.
Given a codeword \(c \in C\), we construct a binary function by outputting the value of the codeword at the index of the input
\(f_c : \{0,1,\dots,2^m-1\} \to \{0,1\},  f_c(x) = c_x\).

We thus consider the set of GSAS given by $\gamma = \{\mathcal{G}(f_c) \mid c \in C\}.$ Note \(\gamma\) is a strict subset of the set of all 2p0s-GSAS, and so if \(\gamma\) does not have an \(\epsilon\)-NE bit-compact representation, neither does the set of all 2p0s-GSAS. As  \(\lvert \gamma \rvert =\lvert C \rvert \geq 2^{(1-H_2(\delta))2^{m}}\), but there are only \(2^{\mathrm{poly}\left(m+2, \log\left(\frac1\epsilon\right)\right)}\) possible bit-compact representation bit strings, there exist games \(\mathcal{F},\mathcal{K} \in \gamma \text{ such that } \mathcal{F} \neq \mathcal{K} \text { but } x^{\mathcal{F}} = x^{\mathcal{K}} =: x\).

\textbf{Cannot be an \(\epsilon\)-NE in both games.}
We proceed by showing that the strategy played by \(R\) using the bit string \(x\) cannot be an \(\epsilon\)-NE for both \(\mathcal{F}\) and \(\mathcal{K}\).

Let \(\pi_2^\mathcal{G}\) be the single valid strategy for player 2 in game \(\mathcal{G} \in \gamma\) (noting that any game in \(\gamma\) has a single action choice available for player 2 at all times).

Let \(\pi_R\) be the strategy output by \(R(S_1,x)\), and \(\pi_{R'}\) the strategy that plays \(H\) whenever \(\pi_R\) plays \(\{1,2,\dots,m\}\) but is otherwise identical to \(\pi_R\). 
As actions \(\{1,2,\dots,m\}\) are strictly dominated by \(\{H\}\), \(U(\pi_{R'}, \pi_2) \ge U(\pi_R, \pi_2)\).
We can thus assume without loss of generality that actions \(\{1,2,\dots,m\}\) are never played by \(R(S_1,x)\).

Let \(\pi^{\mathcal{G}}_1(S_1) = \begin{cases}H &\text{ if }f_{\mathcal{G}}(\bbone_{S_1})=1 \\ T &\text{ if }f_{\mathcal{G}}(\bbone_{S_1})=0\end{cases}\) be the strategy for player 1 in a GSAS \(\mathcal{G} \in \gamma\) which plays the element that, by definition of \(\rho\), player 2 is forced to play. Moreover, note that \(U^\mathcal{\mathcal{G}}_1(\pi_1^\mathcal{\mathcal{G}}, \pi_2^\mathcal{G}) = 1\). Since player 2 cannot deviate from the single valid strategy, and \(1\) is the maximal game value, \(\pi^{\mathcal{G}}\) is a NE.

Then, since by definition \(R\) outputs an \(\epsilon\)-NE,
\begin{align}
	U^\mathcal{F}_1(\pi^\mathcal{F}_1, \pi^\mathcal{F}_2) - \epsilon &\leq U^\mathcal{F}_1(\pi_R, \pi^\mathcal{F}_2)\\
	1 - \epsilon &\leq U^\mathcal{F}_1(\pi_R, \pi^\mathcal{F}_2)\\
	U^\mathcal{F}_1 (\pi_R, \pi_2^\mathcal{F}) &= \mathbb{E}_{S \sim \rho} [U^\mathcal{F}_1(R(S_1, x), \pi_2^\mathcal{F}(S_2))]\\
	&= \sum_{S \in \mathcal{S}} \Pr[S] \cdot U^\mathcal{F}_1(R(S_1,x), \pi_2^\mathcal{F}(S_2)).\label{eqn:nocompact_utility_expectation}
    \end{align}

By definition of \(\gamma\) and the corresponding \(\rho\) for any game (recalling that \(\pi_2^\mathcal{F}\) is restricted to a single output determined by \(f_\mathcal{F}\) and \(\pi_R\) never plays actions \(\{1,2,\dots,m\}\)) we can expand \(U_1^\mathcal{F}\) from Equation~\eqref{eqn:nocompact_utility_expectation} to get
\begin{align}
U^\mathcal{F}_1 (\pi_R, \pi_2^\mathcal{F}) &= \frac{1}{2^m} \sum_{S_1 \in \mathcal{S}_1} \big\{ R(S_1, x)_H \cdot f_\mathcal{F}(\bbone_{S_1}) + R(S_1, x)_T \cdot (1- f_\mathcal{F}(\bbone_{S_1}))\big\}\\
1 - \epsilon &\leq \frac{1}{2^m} \sum_{S_1 \in \mathcal{S}_1} \big\{ R(S_1, x)_H \cdot f_\mathcal{F}(\bbone_{S_1}) + R(S_1, x)_T \cdot (1- f_\mathcal{F}(\bbone_{S_1}))\big\}.\label{eqn:nocompact_expanded_function_form}
\end{align}

We now partition player 1’s availability sets \(\mathcal{S}_1\) into two disjoint sets: \(S_=\), where the binary functions that generated \(\mathcal{F}\) and \(\mathcal{K}\) are equal, using the availability set as a Boolean indicator vector as input (cf. Equation~\eqref{eqn:nocompact_rho}), and \(S_{\neq}\), where the binary functions differ. Formally, we have
\begin{align}
    S_= &= \{S_1 \in \mathcal{S}_1 \mid f_\mathcal{F}(\bbone_{S_1}) = f_\mathcal{K}(\bbone_{S_1})\}\\
    S_{\neq} &= \{S_1 \in \mathcal{S}_1 \mid f_\mathcal{F}(\bbone_{S_1}) \neq f_\mathcal{K}(\bbone_{S_1})\}.
\end{align}

Note that $S_= \cup S_{\neq} = \mathcal{S}_1,  S_= \cap S_{\neq} = \varnothing$. Splitting Equation~\eqref{eqn:nocompact_expanded_function_form} to separately sum over these sets gives
\begin{align}
1 - \epsilon &\leq \frac{1}{\lvert \mathcal{S}_1 \rvert} \sum_{S_1 \in S_=} \big\{ R(S_1, x)_H \cdot f_\mathcal{F}(\bbone_{S_1}) + R(S_1, x)_T \cdot (1- f_\mathcal{F}(\bbone_{S_1}))\big\} \nonumber\\
&\quad+ \frac{1}{\lvert \mathcal{S}_1 \rvert} \sum_{S_1 \in S_{\neq}} \big\{ R(S_1, x)_H \cdot f_\mathcal{F}(\bbone_{S_1}) + R(S_1, x)_T \cdot (1- f_\mathcal{F}(\bbone_{S_1}))\big\}.\label{eqn:nocompact_split_sets}
\end{align}

As payoffs are bounded by 1,
\begin{align}
    \frac{1}{\lvert \mathcal{S}_1 \rvert} \sum_{S_1 \in S_=} \big\{ R(S_1, x)_H \cdot f_\mathcal{F}(\bbone_{S_1}) + R(S_1, x)_T \cdot (1- f_\mathcal{F}(\bbone_{S_1}))\big\} &\leq \frac{\lvert S_= \rvert}{\lvert \mathcal{S}_1 \rvert},\label{eqn:nocompact_sequal_bound}
\end{align}
which, combining with Equation~\eqref{eqn:nocompact_split_sets}, gives
\begin{align}
1 - \frac{\lvert S_= \rvert}{\lvert \mathcal{S}_1 \rvert} - \epsilon &\leq \frac{1}{\lvert \mathcal{S}_1 \rvert} \sum_{S_1 \in S_{\neq}} \big\{ R(S_1, x)_H \cdot f_\mathcal{F}(\bbone_{S_1}) + R(S_1, x)_T \cdot (1- f_\mathcal{F}(\bbone_{S_1}))\big\}.\label{eqn:nocompact_s_notequal_bound}
\end{align}

Then, as \(R(S_1,x) \in \Delta(\{H,T\})\), we have \(R(S_1,x)_H + R(S_1,x)_T = 1\), and since \(f_\mathcal{F}\) is always either $0$ or $1$, we get
\begin{align}
\lvert S_{\neq} \rvert &= \sum_{S_1 \in S_{\neq}} \big\{ R(S_1, x)_H \cdot f_\mathcal{F}(\bbone_{S_1}) + R(S_1, x)_T \cdot (1- f_\mathcal{F}(\bbone_{S_1}))\nonumber\\
&\quad\quad\quad+ R(S_1, x)_H \cdot (1 - f_\mathcal{F}(\bbone_{S_1})) + R(S_1, x)_T \cdot f_\mathcal{F}(\bbone_{S_1})\big\}.
\end{align}

Splitting the sum and substituting into Equation~\eqref{eqn:nocompact_s_notequal_bound} gives
\begin{align}
1 - \frac{\lvert S_= \rvert}{\lvert \mathcal{S}_1 \rvert} - \epsilon &\leq \frac{\lvert S_{\neq} \rvert}{\lvert \mathcal{S}_1 \rvert} - \frac{1}{\lvert \mathcal{S}_1 \rvert} \sum_{S_1 \in S_{\neq}} \big\{ R(S_1, x)_H \cdot (1 - f_\mathcal{F}(\bbone_{S_1})) + R(S_1, x)_T \cdot f_\mathcal{F}(\bbone_{S_1})\big\}\\
\epsilon &\geq \frac{1}{\lvert \mathcal{S}_1 \rvert} \sum_{S_1 \in S_{\neq}} \big\{ R(S_1, x)_H \cdot (1 - f_\mathcal{F}(\bbone_{S_1})) + R(S_1, x)_T \cdot f_\mathcal{F}(\bbone_{S_1})\big\}.
\end{align}

Since in \(S_{\neq}\) we have \(f_\mathcal{K} \neq f_\mathcal{F}\), it follows that
\begin{align}
\epsilon &\geq \frac{1}{\lvert \mathcal{S}_1 \rvert} \sum_{S_1 \in S_{\neq}} \big\{ R(S_1, x)_H \cdot f_\mathcal{K}(\bbone_{S_1}) + R(S_1, x)_T \cdot (1 - f_\mathcal{K}(\bbone_{S_1}))\big\}.\label{eqn:nocompact_snotequal_eps_bound}
\end{align}

Since the availability sets are defined by truth tables read from two binary strings from \(C\), and \(\min_{c_1, c_2 \in C; c_1 \neq c_2} d_{\text{Hamming}}(c_1, c_2) \geq \delta 2^m\), \(f_\mathcal{F}\) and \(f_\mathcal{K}\) must differ on at least a \(\delta\) fraction of the inputs. Thus \(\lvert S_{\neq}\rvert/\lvert \mathcal{S}_1 \rvert \geq \delta \implies \lvert S_=\rvert/\lvert \mathcal{S}_1 \rvert \leq 1 - \delta\). Adding the payoff component contributed by both \(S_=\) from Equation~\eqref{eqn:nocompact_sequal_bound} and \(S_{\neq}\) from Equation~\eqref{eqn:nocompact_snotequal_eps_bound} gives
\begin{align}
1 - \delta + \epsilon & \geq \frac{1}{\lvert \mathcal{S}_1 \rvert} \sum_{S_1 \in \mathcal{S}_1} \big\{ R(S_1, x)_H \cdot f_\mathcal{K}(\bbone_{S_1}) + R(S_1, x)_T \cdot (1 - f_\mathcal{K}(\bbone_{S_1}))\big\}\\
1 + \epsilon - \delta &\geq \mathbb{E}_{S \sim \rho_\mathcal{K}} [U^\mathcal{K}_1(R(S_1,x), \pi_2^\mathcal{K}(S_2))]\\
U^\mathcal{K}_1(\pi^\mathcal{K}_1, \pi_2^\mathcal{K}) + \epsilon - \delta &\geq U^\mathcal{K}_1(\pi_R, \pi_2^\mathcal{K}).
\end{align}

As by definition \(\delta > 2\epsilon\) and by the Gilbert-Varshamov bound $\delta < 1/2$, it follows that
\begin{align}
U^\mathcal{K}_1(\pi^\mathcal{K}_1, \pi_2^\mathcal{K}) - \epsilon &> U^\mathcal{K}_1(\pi_R, \pi_2^\mathcal{K}),
\end{align}
where $\epsilon <\frac{1}{4}$. 
Since \(R\) must be able to play an \(\epsilon\)-NE on \(\mathcal{K}\), we have a contradiction.
\end{proof}

\end{document}